\documentclass[12pt]{article}
\usepackage{microtype}

\usepackage{setspace} 
\usepackage{pstricks}
\usepackage{pst-node}
\usepackage{titlesec}
\usepackage[all]{xy}
\usepackage{enumerate}
\usepackage{graphicx}
\usepackage[margin=1in]{geometry}
\usepackage{pgf}
\usepackage{tikz}
\usepackage{pgfplots}
\usepackage{caption}
\usepackage{subcaption}
\usepackage{amsmath}
\usepackage{amsfonts}
\usepackage{amssymb}
\usepackage{amsthm}
\usepackage{mathrsfs}
\usepackage{natbib} 
\usepackage{pgfplots}
\pgfplotsset{compat=1.18}
\usepackage{tocbibind}
\usepackage{subcaption}
\usepackage[toc,page]{appendix}
\usepackage[colorlinks = true,
linkcolor = blue,
urlcolor  = blue,
citecolor = blue,
anchorcolor = blue]{hyperref}

\usetikzlibrary{decorations.markings}

\newtheorem{prop}{Proposition}

\newtheorem{observation}{Observation}
\newtheorem{example}{Example}
\newtheorem{definition}{Definition}

\newcommand{\supp}{\text{supp}}
\newtheorem{appxlem}{Lemma}[section]
\newtheorem{appxprop}{Proposition}[section]
\newtheorem{appxdef}{Definition}[section]

\newlength\tindent
\begin{document}
	
	\title{Mechanism Design with Endogenous Perception\thanks{We are grateful to the editor, Eduardo Azevedo, and two referees whose insightful comments greatly improved the paper. We also thank Rani Spiegler, Stephen Morris, Juan Carlos Carbajal, Kentaro Tomoeda, Lu{\'i}s Pontes de Vasconcelos, Alexander Jakobsen, Isa Hafalir, Jun Zhang, Antonio Rosato, participants of the UNSW-UQ Economic Theory Festival 2022, participants at the Australasian Meeting of the Econometric Society 2023 and participants at the Behavioral Mechanism Design Workshop 2024 at UNSW for insightful comments.} }
	
	\author{Benjamin Balzer\thanks{University of Technology Sydney (benjamin.balzer@uts.edu.au)} \quad \& \quad Benjamin Young\thanks{University of Technology Sydney (benjamin.young@uts.edu.au)}}
	\date{\today}  
	
	\maketitle
	
	\begin{abstract}
	We model endogenous perception of private information in single-agent
screening problems, with potential evaluation errors. The agent’s
evaluation of their type depends on their cognitive state: either attentive (i.e.,
they correctly perceive their type) or inattentive (i.e., they might misperceive their
type). The mechanism’s incentives structure determines the agent's cognitive state via costly
investment in cognition. We derive a general representation of attention incentives, show how they vary with the mechanism’s allocation rule, and define a notion of accuracy of perception. In applications we showcase how perception both shapes and is shaped by the design of mechanisms.

	\end{abstract}

	\section{Introduction}
	
	There is a long tradition in economics of assuming that individuals hold correct perceptions of their preferences and beliefs. Recently, the study of behavioral economics has suggested that the processing of information about one's own preferences may be subject to bias.\footnote{Examples of processing errors include distorting probabilities when performing expected-utility calculations (see, e.g., \cite{kt1979}, \cite{prelec1998}, and \cite{quiggin2012}), responses to salience effects \citep{gs2010}, or the under- or over-reaction to information \citep{Phillips1966,Barberis1998}.} Moreover, such misperceptions can adversely affect the functioning of markets.\footnote{For example, misperception impacts on gender imbalances in career progression \citep{niederle2007,exley2022},
		costly usage of the court system \citep{goodman2010}, investment in human capital based on perceptions of motherhood \citep{kuziemko2018}, job search \citep{adams2023}, and participation in markets for extended warranties \citep{abito2019}.} Consequently, investigating the impact of such behavioral or misperceived preferences on the performance of institutions has received significant interest. There is evidence, however, that such misperceptions may not constitute a fixed characteristic of individuals but, rather, depend on an individual's cognitive state which itself responds to the environment.\footnote{For example, \cite{kahneman2011} describes a dual-process system of intuitive versus contemplative decision-making that individuals can transition within, depending on the context. Moreover, there is evidence that biases may constitute a `rational irrationality' \citep{tirole2002}, attention responds to incentives \citep{gaglianone2022}, stakes play a role in the persistence of biases 
		\citep{Zimmermann2020,Fehr2022}, and depth of reasoning may depend on incentives \citep{ap2016}.}
	In this paper, we provide a portable theory of how preference perception endogenously changes with the incentives prevalent in the environment. As such, our theory speaks to how the perceptions are shaped by, rather than simply affecting, the design of institutions.\\   
	
	Consider a consumer who is in the market for a car. The consumer needs to inspect the car to perceive their valuation of the car and then, based on this perception, decide whether to purchase the car or not. The process of inspecting the car and forming a valuation depends on how prepared they are to do so. Specifically, to properly assess the car, they need to invest cognitive resources to work out the relevant aspects of the car to inspect and the details to pay attention to given their own needs. Invoking such a mindful process is costly. Instead, the consumer can inspect the car without this cognitive investment, in which case they will only pay attention to salient aspects of the car (i.e., its appearance, its size, etc) while ignoring details of less salient aspects (e.g., mileage, engine quality, etc). These neglected details no longer come to mind \citep{gs2010}, leading the consumer to a biased evaluation of the car's value.\footnote{Effectively, the consumer treats these neglected aspects as \emph{unknown unknowns} as described in \cite{kahneman2011}. As such, they do not question where their perception of the car's value comes from when making their purchase decision.} The consumer's incentives to exert cognitive effort to ensure an accurate evaluation process depend on how costly inattentive information processing is in terms of making decisions based on a biased valuation. In particular, the mechanism that the seller uses (e.g., the posted price of the car) plays a role in determining whether such cognitive investment is warranted. Indeed, if salient aspects of the car serve to increase its value in the consumer's mind, then higher prices might incentivize greater investment in cognition to avoid costly over-purchase. \\
	 
	In Section \ref{section:model} we introduce our model of endogenous perception in screening problems with private information. We focus on a single agent with quasi-linear utility. Such settings include, but are not limited to, efficient provision of goods, maximizing revenue through price discrimination, and
	optimal incentive contracts with adverse selection. In these settings, a designer aims to allocate resources optimally given their objective. Specifically, the designer chooses a rule that determines the agent's allocation as a function of their type report. This allocation rule has a different interpretation depending on the context. For example, it might specify whether the agent receives a good, the quality of provision of a good, or the level of effort the agent should exert. The designer faces the friction that they lack complete knowledge of the agent's privately-known type. As such, the first-best allocation rule is typically not feasible and the designer has to screen the agent. That is, the designer trades off providing rents to the agent to reveal their private information against the implied costs on the objective.  \\

    We consider this simple environment in order to focus on the implications of endogenous perception formation in mechanism design problems. Indeed, our main contribution is to introduce the idea that the agent forms their type perception in different cognitive states, where the state the agent operates in responds to the mechanism's incentives. Specifically, we separate the agent's decision-making problem into an evaluation phase (where the agent forms a perception of their type) and a mechanism phase (where the agent takes optimal actions given this perception). The agent is in one of two cognitive states: either attentive (i.e., always perceive their type correctly) or inattentive (i.e., potentially misperceive their type). In the evaluation phase, when inattentive, the agent forms their type via a \emph{perception-generating process} (PGP), mapping true types into perceptions of these types. Our framework allows for arbitrary PGPs, including both unbiased ones (i.e., generated via information) and biased ones. For example, we allow for both optimism (i.e., perception always being higher than the true type) and pessimism (i.e., perception always falling below the true type), which are inconsistent with information processing. In the mechanism phase, the agent takes their formed perception as given (i.e., treats their perception as their true type) and chooses an optimal action accordingly. Whether the agent is attentive or inattentive when evaluating their type depends on whether they exert sufficient cognitive effort during this process. In particular, the agent is attentive if the expected payoff difference between what the agent obtains from being attentive versus operating under inattention exceeds some cognitive cost. We call this payoff difference the \emph{value of attention}. Incentives to make a cognitive investment are tied to the mechanism's allocation rule, as this crucially shapes the value of attention. Hence, when choosing a mechanism, the designer needs to consider not only incentives for truthful reporting of types but also incentives for the agent's attention decision.\\

   Our preferred interpretation of the model is one in which the agent is not aware of any biases that realize but is \emph{meta-aware} about their potential to form biases and, thus, can avoid doing so at a cognitive cost.\footnote{We borrow the term \emph{meta awareness} from \cite{Loewenstein2003}. The idea that individuals are only meta-aware rather than aware of their biases is a common theme in the behavioral economics literature. In Section~\ref{sec:discussion}, we further motivate this interpretation of our model.} We employ a `dual-selves' approach to micro-found this meta-awareness: the agent consists of a planner and an inattentive doer. Before the decision-making process, the planner decides whether to be attentive or inattentive. If attentive, the planner takes control of the decision-making process and perceives their type correctly. If inattentive, however, the planner delegates decision making to the inattentive doer, who may arrive at an inaccurate perception of the true type.\footnote{There are alternative interpretations of our model. For example, one can think of the planner as an organization that is aware that their employees are subject to biases in evaluation. De-biasing the employees, through, for example, mandatory Unconscious Bias Training, will be offered only if the benefits of unbiased decision-making outweigh the costs. Similarly, a benevolent and rational third party (a friend, parent, sibling, teacher or advisor) can intervene in the agent's process by providing advice. Costs then stem from either the time involved or the emotional costs of bearing bad news to a person one likes. Finally, one can interpret our model as evolutionary. If the agent faces similar decisions over time, they only become aware of their biases if their past outcomes were sufficiently adverse.} When determining the cognitive state, the planner takes the inattentive doer's behavior, and thus the consequences of arriving at a potentially biased perception, into account. They then choose their cognitive state by trading off the utility losses from inattention against cognitive costs of paying attention.\footnote{Instead of being purely cognitive, these costs could also stem from missed opportunities. For example, a consumer might need to engage in time-consuming research about important aspects of cars before inspection, or they need to satiate their hunger before going shopping. Additionally, these costs could come from monetary payments. For instance, prospective buyers might hire experts to inspect the house structure before making bids, or a financial investor might subscribe to an investment research platform to better evaluate opportunities.} \\

    Since the agent's value of attention determines their cognitive state, we explore how this value depends on the mechanism's features in Section~\ref{sec:VOS}. First, we show that the value of attention admits a representation as a linear functional of the allocation rule. This representation makes it tractable to incorporate endogenous perception via attention constraints into mechanism design problems. Next, we establish that screening is necessary to create attention incentives and that all maximizers of the value of attention share the property that they pool sufficiently high types at maximal allocation and sufficiently low types at minimal allocation. 
	Such coarse screening maximizes the value of attention independently of the agent's PGP as it exacerbates the difference in the behavior of an inattentive agent to that of their attentive counterpart. As such, we highlight a novel cost of screening: providing strong incentives for attention demands coarseness in the allocation rule.\\
	
	Finally, we use the value of attention to define a notion of \emph{accuracy} of perception. Specifically, we say that one PGP is more accurate than another if the former grants the agent higher welfare in every mechanism for any cognitive cost than the latter. We provide a characterization of accuracy based on the statistical properties of the agent's PGP only. This sufficient statistic has a natural interpretation as a measure of (i) how much information is contained in the agent's PGP and, (ii) the degree to which this PGP is biased. We showcase a number of PGPs that are ordered by accuracy. In particular, we show that informativeness (in the sense of mean-preserving spreads) is equivalent to accuracy for unbiased perceptions, that this order captures natural notions of increasing bias in perception, and that some PGPs are ordered even when they conflict in terms of their informational content and degree of bias. \\
	
	We apply our framework in three settings to explore the role that both biased and endogenous perception plays in shaping optimal mechanisms. First, we revisit the problem of designing costly provisions to maximize efficiency. In this setting, it is well-known that ``selling the firm to the agent" (i.e., having the agent internalize the costs of production) is optimal if perception is unbiased. Instead, a biased agent requires the production process to be managed as they do not correctly internalize their own contribution to welfare. Moreover, we establish that even if efficiency is achievable when the agent's cognitive state is fixed, this may not be the case when this state is endogenous. In particular, we identify the novel issue of \emph{over-attentiveness:} if managing the process to account for the agent's biases requires screening, it may provide the agent with too strong incentives to pay attention, thus rendering efficiency unachievable. \\
	
	Second, we explore the role of the PGP in shaping how a revenue-maximizing designer optimally screens as in \cite{Mussa1978}. There, we consider two PGPs that generate the same distributions of both inattentive and attentive perceptions, however, via distinct processes. In contrast to the case of information where, in general, only the final distribution of beliefs matters, we highlight that the entire process via which perception is formed is integral to the design of optimal mechanisms. Specifically, we show a designer who desires the agent to pay attention does so, depending on the PGP, by either using a \emph{carrot} (i.e., increasing the benefits of attention relative to inattention) or a \emph{stick} (i.e., decreasing the benefits of inattention relative to attention). \\
	
	Finally, we investigate the impact of hype in product markets. We consider a buyer who may fall prey to hype if inattentive, in which case they purchase a product independently of their type. Hype is beneficial for welfare as it encourages efficient trade. However, if the agent's cognitive state is fixed, the seller has incentives to extract too much of these welfare gains as rent and, consequently, the buyer prefers not to succumb to hype. Instead, when the cognitive state is endogenous, we show that the buyer may prefer a strictly positive degree of hype. This is due to a novel commitment-flexibility trade-off afforded by endogenous but biased perception: succumbing to hype 
	affords commitment to purchase the product, thus increasing welfare, while flexibility to potentially pay attention and avoid hype limits the seller's possibility to extract rents, allowing the buyer to reap some of these welfare gains. 
	
	\section{Literature Review}
	
	There is a large literature on mechanism-design problems with behavioral agents. Examples include agents who hold dynamically inconsistent preferences \citep{es2006,galperti2015,yu2020}, who are reference-dependent and loss-averse \citep{ce2016,brv2022, gershkov2022}, who are optimistic or over-confident \citep{es2008,sautmann2013}, who exhibit social preferences \citep{kucuksenel2012,imas2022}, who are unaware \citep{vonThadden2012}, who are ambiguity averse \citep{tillio2016}, who struggle to understand the mechanism \citep{gr2014,jakobsen2020}, who exhibit taste projection \citep{antonio2021}, or who are not strategically sophisticated \citep{li2017,pt2019,dl2021}. These papers take the agents' behavioral preferences or bounds on rationality as given and ask how institutions are then shaped by them. Instead, we allow the mechanism to affect the propensity of agents to misperceive their preferences. Thus, our work sheds light on how institutions should be designed to ameliorate (rather than simply account for) the biased decision-making of participants. \\

	There are a number of papers that investigate endogenous perception of preferences and beliefs in strategic environments. These papers generally model perceptions as garblings and thus are models of endogenous information. Within that class we distinguish between two literature strands. The first strand focuses on how much information is made available in given strategic environments  (e.g., see \cite{roesler2017, bergemann2017,bergemann2019}). The second strand concerns how much information is optimally acquired or processed (e.g., see seminal work by \cite{bergemann2002,shi2012} and more recent work by \cite{mensch2022,ravid2020,ravid2022}). 
	In contrast, our model is not one in which the agent acquires information but, rather, investigates the extent to which they can avoid (potentially biased) evaluation of available private information. In this sense, our model shares a similar interpretation to those of rational inattention. In rational inattention, beliefs are generated via garblings (so that perception is unbiased) and these models usually allow for arbitrary garblings. While our model admits unbiased perception as a special case, our focus is on the behavioral notion of misperception and, thus, we allow for more general forms of perception than are admitted by garblings. This comes at the cost of less generality in terms of our restriction to binary cognitive states. For example, we allow for optimism or pessimism at all types; a systematic form of misperception that rational models of information acquisition cannot capture due to the requirement that posterior beliefs under garblings average to the prior. Finally, a few recent contributions explored how biased perceptions can form in market settings \citep{bridet2024, schwardmann2019, immordino2011, immordino2015}. In these papers the agent demands biases, while our agent incurs costs avoiding biases. Instead, the closest paper to ours is \cite{young2022}, who explores endogenous overoptimism in a market setting without private information. We generalize his framework by incorporating private information and a broader range of perceptions, nesting the agent's decision problem in \cite{young2022}. This allows us to more thoroughly examine the interaction between attention incentives and (the design of) mechanisms.  
	\section{The Model}\label{section:model}
	
	\subsection{Formal Framework}

	\textbf{Basic Setting:} We consider a class of canonical screening problems with a single risk-neutral agent. The agent faces a mechanism, $\mathcal{M}$, and takes an action which, according to the rule of the mechanism, generates an outcome. The outcome of the mechanism is a tuple $(q, t)$ where $q \in [0,1]$ is the probability or amount of allocation and $t \in \mathbb{R}$ is a transfer. The agent has private type $\theta \in [0,1]$ drawn according to distribution function $F$. An agent with type $\theta$ receives utility $q \theta - t$ from outcome $(q, t)$.\footnote{The restrictions that $q, \theta\in [0, 1]$ are for ease of exposition. These can be relaxed to fall in arbitrary compact subsets of $\mathbb{R}$ without affecting the main findings.} \\
	
	\textbf{Model Overview:} Our model is one in which the agent may not correctly perceive their true type, $\theta$, but has control over the propensity to do so via their cognitive state. We model the agent's decision-making process in two phases, followed by the initial determination of the cognitive state: (i) an evaluation phase, and (ii) a mechanism phase. Depending on their determined cognitive state, the agent arrives at a perception of their type, $\pi \in [0, 1]$, in the evaluation phase, which is a function of their true type, $\theta$. In the resulting mechanism phase, the agent takes their perception as given (i.e., treats $\pi$ as their true type) and chooses an optimal action given mechanism $\mathcal{M}$. Outcomes are then determined according to the rules of the mechanism as well as how the agent plays $\mathcal{M}$.\\

    Figure~\ref{fig:T} summarizes the model. The interpretation of this model is that of an agent who is made up of two selves: a planner and an inattentive doer. The planner chooses their cognitive state; either attentive (denoted by $A$) or inattentive (denoted by $I$). If attentive the planner takes control of the decision-making process. In particular, during the evaluation phase, they correctly perceive their private information to be $\theta$. If inattentive, however, the planner delegates control to the inattentive doer (e.g., the agent operates on `autopilot' or has distinct objectives that depend on, say, mood; e.g., being hungry) who forms perception of private information that does not necessarily align with the truth (i.e., $\pi$ need not equal $\theta$). Hence, the planner can guarantee optimal behavior in the mechanism only if they are in full control of the decision-making process. Paying attention in this way is costly, and thus creates a trade-off for cognitive state determination.\\

    \begin{figure}  
\centering
     \scalebox{0.7}{ 
	\begin{tikzpicture}[node distance=2cm]
		
 		\node[draw, circle, minimum width=0.025cm, minimum height=0.025cm, fill] (start) at (-0.5,0) {};
		\node[draw, minimum width=2cm, minimum height=1.5cm,blue] (box1) at (3.875,1.6) {
			\begin{tabular}{c}
				\small{{$\theta$ realises}}\\
				\small{{and $\pi=\theta$}}
			\end{tabular}
		};
		\node[draw, minimum width=2cm, minimum height=1.5cm,red] (box2) at (3.875,-1.9) {
			\begin{tabular}{c}
				\small{$\theta$ realizes and} \\
				\small{ $\pi$ drawn} \\
				 \small{with CDF $\rho(\pi|\theta)$}
			\end{tabular}
		};
		\node[draw, minimum width=2cm, minimum height=1.5cm,blue] (box3) at (8.1,1.75) {
			\begin{tabular}{c}
			\small{Agent plays $\mathcal{M}$   }\\ \small{ optimally,}\\
		\small{ taking $\pi$}\\ \small{as true type.}   
			\end{tabular}
		};
		\node[draw, minimum width=2cm, minimum height=1.5cm,red] (box4) at (8.1,-1.75) {
			\begin{tabular}{c}
				\small{Agent plays $\mathcal{M}$   }\\ \small{ optimally,}\\
			\small{ taking $\pi$}\\ \small{as true type.}  
			\end{tabular}
		};
	
% \draw[dashed,blue] (6.7,1.7) ellipse (6 cm and 1.5cm);
%  \draw[dashed,red] (6.7,-1.8) ellipse (6 cm and 1.5cm);
 \draw[thick,dashed] (1.75,0)--(12,0);
		
		\draw[thick,->] (start) -- (box1);
		\draw[thick,->] (start) -- (box2);
		\draw[thick,->] (box1) -- (box3);
		\draw[thick,->] (box2) -- (box4);
		\draw[dashed,thick] (1.75,-3.5)--(1.75,3);
			\draw[dashed,thick] (10.2,-3.5)--(10.2,3);
			\draw[dashed,thick] (6,-3.5)--(6,3);
		 \node at (0,-3.75)   {Cognitive State};
          \node at (0,-4.25)   {Determination};
          \node at (3.875,-3.75)   {Evaluation};
          \node at (3.875,-4.25)   {Phase};
         \node at (8.1,-3.75)   {Mechanism};
          \node at (8.1,-4.25)   {Phase};
		  \node[blue] at (11.5,1.5)   {Planner};
        %  \node[blue] at (14,1.5)   {Doer};
            \node[red] at (11.5,-1.5)  {Inattentive};
          \node[red] at (11.5,-2)   {Doer};
            \node at (0.5,0.75)   {\textcolor{blue}{$A$}};
             \node at (0.5,-0.75)   {\textcolor{red}{$I$}};
             \node at (-0.75, 0) [left,blue] {Planner};
	\end{tikzpicture}}

		\caption{\small{The Structure of the Model.}}
		\label{fig:T}
\end{figure}  
	
	\textbf{Evaluation Phase:} The evaluation phase determines how the agent perceives their type. This realized perception, $\pi \in [0, 1]$, depends both on the true type, $\theta$,
	and the agent's cognitive state at the point in which they evaluate this type. If in cognitive state $A$ (i.e., the planner controls the decision-making process), the agent correctly perceives their type; that is, $\pi = \theta$. In contrast, in state $I$ (where the inattentive doer is in charge), $\pi$ is drawn according to CDF $\rho(\cdot | \theta)$, which depends on true type $\theta$. We call $\rho \equiv \{\rho(\cdot|\theta): \theta \in [0, 1]\}$ the inattentive doer's \emph{perception-generating process (PGP)}. We place no \emph{a priori} structure on $\rho(\cdot | \theta)$, thus allowing for arbitrary misperception of $\theta$.\footnote{We take an, admittedly, reduced-form approach to modeling the formation of perception. That is, we abstract from the process that determines perception and focus instead on the outcome of this process (i.e., $\rho$). By doing so, our model incorporates both deterministic and stochastic perceptions, with unbiased perception (i.e., information) as a special case. See Section~\ref{sec:examples} for examples of prominent models of misperception our model nests.} Upon forming their perception, each agent behaves as if this perception is their true type.\\

	\textbf{Mechanism Phase:} We focus on direct-revelation mechanisms that incentivize the agent to truthfully report their perception, $\pi$.\footnote{Note that, in our setting, the revelation principle holds so that this restriction is without loss of generality. This is because the agent is an expected-utility maximizer and their perceived type, which is in $[0, 1]$, is a sufficient statistic for their behavior, regardless of the cognitive state they are in, as they never question it.} 
	With a slight abuse of notation, for each report $\hat{\pi}\in [0, 1]$, mechanism $\mathcal{M}$
	specifies outcome $(q(\hat{\pi}), t(\hat{\pi}))$, where $q:[0, 1] \to [0, 1]$ is an allocation rule
	and $t:[0, 1] \to \mathbb{R}$ is a transfer rule. Direct-revelation mechanism $\mathcal{M}$ is incentive compatible. Let $U(\hat{\pi} | \pi) \equiv \pi q(\hat{\pi}) - t(\hat{\pi})$ denote the agent's perceived utility from reporting $\hat{\pi}$ with perception $\pi$. Then, incentive compatibility is equivalent to $\pi \in \arg \max \limits_{\hat{\pi}\in [0, 1]} U(\hat{\pi} | \pi)$ for all $\pi \in [0, 1]$. It is well-known that incentive compatibility in this setting
	requires that $q(\cdot)$ be non-decreasing and 
	$$
	t(\pi) = \pi q(\pi) - \int \limits_{0}^\pi q(x) dx - U(0|0),
	$$
	for all $\pi \in [0, 1]$. Let $Q$ denote the set of feasible allocation rules.\\
	
	\textbf{Determination of the Cognitive State:} Whether the agent is inattentive or attentive depends on the mechanism. Specifically, the agent's planning self compares the ex-ante expected utility (taken with respect to their true type distribution) from attentive and inattentive behavior. However, attention comes at cognitive cost: the agent incurs cost $\kappa \ge 0$ to be attentive by having their planning self take control of the decision-making process. Instead, delegating decision making to their inattentive doing self comes at no cost. \\
	
	Recall that, in an incentive-compatible mechanism, the agent truthfully reports their perception, $\pi$. Moreover, if attentive $\pi = \theta$, and if inattentive $\pi$ is drawn, conditional on $\theta$, according to CDF $\rho(\cdot|\theta)$. Let $V_A(\mathcal{M}) \equiv \int \limits_0^1 U(\theta | \theta) dF(\theta)$ denotes the agent's utility if attentive (net of cognitive cost $\kappa$) and $V_I(\mathcal{M}) \equiv \int \limits_0^1 \int \limits_0^1 U(\pi | \theta)d\rho(\pi|\theta) dF(\theta)$ denote that under inattention. Then, the agent is willing to be attentive if and only if $V_A(\mathcal{M}) - \kappa \ge V_I(\mathcal{M})$. Equivalently, the agent is willing to be attentive if and only if
	\begin{equation}\label{eqn:vos}
		\nu(\mathcal{M}) \equiv \int \limits_{0}^{1}\int \limits_0^1 [U(\theta | \theta) - U(\pi | \theta)]d\rho(\pi|\theta) dF(\theta) \ge \kappa.
	\end{equation}
	We call $\nu(\mathcal{M})$ the \emph{value of attention}, which is weakly positive in any incentive-compatible mechanism.\footnote{This follows trivially from the fact that $U(\pi | \theta)$ is maximized at $\pi = \theta$ in an incentive-compatible mechanism.} Note that $\nu(\mathcal{M})$ is a sufficient statistic for the role the mechanism plays in determining the agent's cognitive state.\\
	
	Our formulation is rich enough to incorporate both exogenous inattention (when $\kappa = +\infty$) and exogenous attention (when $\kappa = 0$). If, however, $\kappa \in (0,+\infty)$, then the agent's cognitive state depends on the value of attention.\footnote{This is similar in spirit to $(u, v)$-procedures as introduced in \cite{kalaietal2002}. There, an individual maximizes $u$ as long as it yields sufficiently high $v$ value, else they maximize $v$. Here, an individual stays inattentive if it yields sufficiently high utility, measured using their true distribution of types.} As such, an increasing $\kappa$ forms a bridge between exogenous attention and exogenous inattention.

	\subsection{The Informational Content of Perception}
	
	In our model, the inattentive agent's behavior is determined only by their realized perception. However, the attention decision depends on the agent's true distribution of types and their PGP. As such, attention incentives depend not only on perceptions
	when inattentive (i.e., how the agent will behave) but, also, implicitly on the information contained within these perceptions about their true type (i.e., how the agent should behave). We now proceed to define what is meant by this informational content.\\ 
	
	Let $F_I(\pi) \equiv \int \limits_0^1 \rho(\pi | \theta)dF(\theta)$ denote the distribution
	of the agent's perception when inattentive. We let $\supp(F_I)$ denote the support of distribution function $F_I$.\footnote{Formally, 
		$$
		\supp(F_I) \equiv \left\{\pi \in [0, 1]: \text{$\int \limits_{\pi-\epsilon}^{\pi + \epsilon} dF_I(\pi) > 0$ for all $\epsilon > 0$}\right\}.
		$$} Let $\rho(\cdot | \pi)$ denote \emph{some} conditional
	distribution of true type $\theta$ given perception $\pi$ generated under inattention. Note that, in our setting, some conditional distribution function $\rho(\cdot | \pi)$ always exists for all $\pi \in [0, 1]$ and is uniquely determined almost everywhere on $\supp(F_I)$. Then, let 
	$$
	e_I(\pi) \equiv \int \limits_0^1 \theta d\rho(\theta | \pi)
	$$
	denote the conditional expectation of the agent's true type, $\theta$, given that perception $\pi$ is generated using PGP $\rho$. Since the agent's utility is linear in their type, $e_I(\cdot)$ is a sufficient statistic for the informational content of the agent's perception in that it captures how the agent would behave if they took $\pi$ as a signal of their true type. However, our agent takes their perception as the best estimate of their true type and, as such, behaves according to $\pi$. Thus, when $e_I(\pi) \ne \pi$, there is a wedge between how the agent \emph{should} behave and how they \emph{actually} behave. We define the notions of unbiased and biased perceptions as follows: 
	
	\begin{definition}\label{defn:unbiased}
		Let $\rho$ denote the agent's PGP. We say $\rho$ is \textbf{unbiased} if $e_I(\pi) = \pi$ almost everywhere on $\supp(F_I)$. Else, we say that $\rho$ is \textbf{biased}.
	\end{definition}

	\subsection{Discussion of Key Modeling Assumptions}\label{sec:discussion}
We consider an agent who is meta-aware of their potential to be biased in their evaluation of information. We micro-found the agent's meta-awareness via a dual-selves approach: there is a planner who decides whether to avoid misperception at cognitive cost $\kappa$, or delegate to an inattentive doer who may misperceive their type. We now discuss this modeling approach.\\

Both for the planner and the doer, we separate the decision-making process in two phases: an evaluation phase followed by a mechanism phase. Our main focus in this paper is to understand how an agent's perception of their private information depends on the design of mechanisms. As such, we assume that the agent understands how the mechanism generates outcomes as a function of their behavior, independent of their cognitive state. Thus, while our agent may misperceive their own type, they correctly perceive the mechanism's rules. The potential for the agent to misperceive the rules of the mechanism due to complexity is well-studied (see \cite{gr2014,li2017,pt2019,jakobsen2020,dl2021}). While we abstract from such complexity issues in order to isolate the impact of a mechanism's incentive structure on how private information is perceived, behavior may still depend on misperceived private information which arises during the evaluation phase if the inattentive doer is in charge. 

\subsubsection{The Inattentive Doer}

 When the planner enters the inattentive cognitive state, the agent's beliefs form according to a fixed process, beyond the planner's control, which we consider a model primitive. Consequently, our notion of perception is in reduced form; i.e., we do not explicitly model the process by which the inattentive doer forms their perception. This approach allows for various notions of perception of private information to be embedded into mechanism design. Examples, some of which are formalized in the next section, include probability weighting (e.g., \cite{prelec1998}), over-confidence (see \cite{Malmendier2015} for a survey) and other notions of motivated beliefs \citep{bt2002,Brunnermeier2005,Caplin2019},\footnote{In \cite{Brunnermeier2005} and \cite{Caplin2019}, the agent's beliefs form endogenously as the resolution of the trade-off between felicity (benefits) and accuracy (costs) in decision making. Our model of binary cognitive states takes the PGP when inattentive as primitive (i.e., we do not explicitly model the benefits of motivated beliefs). Yet, our model is similar in spirit in that the agent's perception forms via a biased process only if the (opportunity) cost of being inattentive (the value of attention) is sufficiently low.} categorical thinking \citep{mullainathan2002} and coarse Bayesian updating \citep{Jakobsen2021}, conservatism bias \citep{Phillips1966,Kovach2021}, under- or over-reaction to news \citep{Barberis1998}, base-rate neglect (see \cite{Kahneman1972} and \cite{Benjamin2019} for a recent survey), responses to uninformative framing \citep{mss2008}, and salience effects \citep{gs2010,bgs2012,bgs2013}. Thus, our framework is portable to a variety of contexts where different aspects of misperception may be more relevant than others.\\

 In the mechanism phase, the inattentive doer takes as given their perception after it has formed, and never questions it.\footnote{Of course, there may be situations in which the agent would be able to question their realized perception. In this case, the agent may go back and evaluate further, arriving at a more refined perception. Such a process would still be represented by some PGP capturing residual misperception (potentially biased) that is unable to be questioned.} One interpretation of this is that the agent only pays attention to certain aspects in their evaluation process, while neglecting others (due to salience effects, for example). Their perception of private information is then based on their recount of what was observed via evaluation. Depending on the agent's PGP, ignored aspects are either \emph{known unknowns} (so that perception is unbiased as in rational inattention) or \emph{unknown unknowns} in that they never come to the agent's mind so that their perception is biased.\footnote{These terms were introduced by \cite{kahneman2011} to describe how an individual updates beliefs on the basis of what they observe.} The agent's perception is then internally consistent with what they recall, even if generated in a biased manner. Indeed, there is recent empirical evidence by \cite{Enke2020} that unobserved factors operate as unknown unknowns, whereby what an agent sees is all they believe there is. Alternatively, the inattentive doer's evaluation procedure might be determined by mood (e.g., being hungry) leading them to evaluate private information differently from how the planner would. This could either be the result of mistakes-- i.e., given the agent's mood, they are unable to exert the necessary cognitive resources to arrive at the truth (and rather employ heuristics)-- or via a misalignment of preferences between the planner and the doer. In the latter case, biases in perception are relative to the planner's preferences.\\

   Of course, in order to endogenize the inattentive PGP, one could extend our framework to allow for more than two cognitive states, ordered via some cognitive hierarchy with increasing cognitive costs as the agent operates in a `higher' cognitive state. In such a case, one could interpret higher cognitive states as the planner more intensively monitoring the inattentive doer. This is easier to formalize for certain classes of misperceptions where a natural cognitive hierarchy exists compared to others where the presence of such an order is less obvious. To maintain portability of our framework, we abstract from modeling specific cognitive hierarchies which are more natural to embed on a contextual basis.

\subsubsection{Determination of the Cognitive State}

    According to our preferred interpretation, before entering the evaluation phase the planner decides whether to control the decision-making process at cognitive cost $\kappa$ or delegate it to the inattentive doer. Thus, our agent appears like a person who is meta-aware of their potential for biases (occurring if the inattentive doer controls the decision-making process) at the time when determining their cognitive state (i.e., when the planner decides whether to control the process). The concept of meta-awareness—i.e., the agent determines their cognitive state but is not aware of biases in the underlying process once the perception has formed—allows us to jointly model that (i) biases in perception potentially play a role in decision making, and (ii) the propensity for such biases responds to the mechanism's incentive structure.\footnote{Whether agents are aware or only meta-aware of their biases is generally an important consideration in behavioral models, as discussed in the conclusion of \cite{Loewenstein2003}. There is, however, anecdotal evidence that people are somewhat meta-aware of their potential biases. For example, with the rise of behavioral science there is an abundance of resources directed towards how to avoid biases in evaluation that people have likely been exposed to. As such, we think it is reasonable that, if the stakes are high enough, people will employ these techniques to avoid bias.} Indeed, while there is evidence of biases in perception and decision making, there is also evidence that their occurrence responds to the incentives present in the choice environment (e.g., see \citet{Zimmermann2020, Fehr2022}).\\  

    Finally, note that our notion of meta-awareness utilizes rational expectations: the planner takes the future behavior of the inattentive doer as given and determines the cognitive state by weighing the objective benefits of cognitively demanding behavior (in our case, avoiding errors when inattentive) against cognitive costs. Whether such a trade-off would be resolved rationally is inherently empirical in nature. However, using rational expectations as a starting point for modeling endogenous perception in mechanism design has advantages. First, it places significant discipline on our model, ensuring both tractability and portability across different contexts. Indeed, it enables us to embed (admittedly coarse) notions of rational inattention (via unbiased PGPs). Second, the idea that an agent's beliefs form in response to incentives via a cost-benefit trade-off is commonly utilized in the behavioral literature, for example, in level-k reasoning \citep{larbi2022}, rational inattention \citep{gaglianone2022}, ambiguity \citep{maccheroni2006}, and optimal expectations \citep{Brunnermeier2005}.

	\subsection{Examples of Perception-Generating Processes}\label{sec:examples}
	
	\textbf{Probability Weighting:} Consider a situation in which the agent's true type, $\theta$,
	represents the probability that they value allocation $q$. Thus, an attentive
	agent values outcome $(q, t)$ according to $\theta q - t$. Instead, the inattentive agent weights probability $\theta$ according to $\pi(\theta)$. Note that this is a deterministic perception where distribution $\rho(\pi | \theta)$ has a Dirac measure with point-mass one at $\pi = \pi(\theta)$. Our model flexibly incorporates any form of probability weighting, including over-optimism (i.e., $\pi(\theta) > \theta$ for all $\theta \in (0, 1)$), under-optimism (i.e., $\pi(\theta) < \theta$ for all $\theta \in (0, 1)$), and mixtures of the two such as that proposed by \cite{prelec1998} (i.e., $\pi(\theta) = e^{-\beta(-\log \theta)^\alpha}$ with $\alpha,\beta>0$ for all $\theta \in (0, 1)$).\\
	
	\textbf{Conservatism Bias:} Suppose that the mean of the agent's true type under $F$ is given by $\mu$. The inattentive agent updates conservatively in the sense that they under-update their type when it realizes by continuing to place weight on prior $\mu$. This can be modeled by assuming that there is an $\alpha \in (0, 1)$ such that $\pi(\theta) = \alpha \theta + (1 - \alpha) \mu$. Note that, as in the case of probability weighting, this is a deterministic perception where distribution $\rho(\pi | \theta)$ has Dirac measure with point-mass one at $\pi = \pi(\theta)$. As $\alpha$ increases, the agent places more weight on their true type, and the degree of conservatism decreases.\\
	
	\textbf{Hype in Product Markets:} There is a product that a buyer has true valuation $\theta \sim F$ for. There is, however, hype surrounding the product. If the buyer is inattentive when evaluating the product, then they succumb to hype with probability $h \in [0, 1]$. If they succumb to hype, they incorrectly perceive their valuation for the product to be maximal (i.e., $\pi = 1$) and, as such, will always purchase the product. If the individual does not succumb to hype, then they correctly evaluate their type (i.e., $\pi = \theta$). It is natural to think that $h$ is increasing in how much hype there is surrounding the product and, as such, can be thought of as the \emph{degree of hype}. \\
	
	\textbf{Information:} In a model of information acquisition or information processing (such as rational inattention), agents are generally assumed to be Bayesian in the sense that they use their perception
	as a signal of their true type. Our model incorporates information as the class of unbiased perceptions (see Definition~\ref{defn:unbiased}). In such situations, the true type distribution, $F$, is a mean-preserving spread of the distribution of
	the agent's perception when inattentive, $F_I$ (\cite{MWG}, Proposition 6.D.2). Thus, our framework incorporates binary models of information in which the agent either acquires maximal information (i.e., is attentive) or minimal information (i.e., is inattentive). \\
	
	\textbf{Fictitious Information:} While our model allows for any notion
	of perception based on information loss, it is sufficiently general to also allow the agent to incorrectly perceive that information has been generated when no such information exists. In particular, suppose that $F(\theta) = 1$ if and only if $\theta \ge 1/2$; that is, the agent's true type is always $1/2$. There is another variable, $\pi$, distributed uniformly on $[0, 1]$ and independent of $F$ (i.e., contains no information regarding $\theta$). The inattentive agent, however, incorrectly perceives that $\pi$ is an unbiased estimate of $\theta$ (i.e., $\rho(\pi | 1/2) = \pi$ for all $\pi \in [0, 1]$). As such, the agent behaves as if they have received information: their perception is uniformly distributed on $[0, 1]$ which is a mean-preserving spread of the true type distribution, $F$. This fits with ideas such as consumer responding to uninformative framing, as in \cite{mss2008}.

	\section{The Value of Attention}\label{sec:VOS}
	
	In this section, we explore the attention incentives of the agent in more detail. Our objective is to understand how the agent's value of attention depends on both the mechanism and their own PGP. We begin by providing a useful representation of the value of attention and describe some necessary aspects of both the mechanism and the agent's PGP for its value to be positive.
	
	\begin{prop}\label{prop:gen Z}
		For a given direct mechanism, $\mathcal{M} = (q, t)$, the value of attention admits the representation
		\begin{equation}\label{eqn:vos}
			\nu(\mathcal{M}) = \int \limits_0^1 q(\pi)[F_I(\pi) - F(\pi)]d\pi + \int \limits_0^1 q(\pi)[\pi - e_I(\pi)]dF_I(\pi).
		\end{equation}
		Moreover, 
		\begin{enumerate}[(a)]  
			\item $\nu(\mathcal{M}) = 0$ for all $\mathcal{M}$ if and only if $\rho$ is unbiased and $F=F_I$; and
			\item  $\nu(\mathcal{M}) = 0$ for any PGP $\rho$ if and only if $q$ is a constant function. 			
		\end{enumerate}
	\end{prop}
	
	Equation \eqref{eqn:vos} provides a tractable representation of the value of attention for a given mechanism. In the spirit of the Revenue Equivalence Principle, the value of attention depends only on the allocation rule.\footnote{Note one subtle difference to the Revenue Equivalence Principle is that the value of attention is \emph{uniquely} pinned down by the mechanism's allocation rule, not simply up to a constant.} Moreover, because the agent's utility is linear in $q$, the value of attention is a linear functional in that allocation rule and, thus, is simple to incorporate (as a constraint) in any mechanism design problem.\\
	
	Representation~\eqref{eqn:vos} consists of two parts. The first is the agent's expected utility difference between operating under true distribution, $F$, and operating under inattentive distribution, $F_I$, assuming that $F_I$ were generated from an unbiased PGP. However, since the PGP may be biased, the second term of \eqref{eqn:vos} corrects for this potential misperception. Consequently, this term is only present in the case of a biased PGP. Indeed, if the PGP is unbiased, $e_I(\pi) = \pi$ almost everywhere on $\supp(F_I)$, so that the second term disappears.\footnote{Since the value of attention is non-negative for all non-decreasing allocation rules, true distribution $F$ is a mean-preserving spread of perceived distribution $F_I$ if the PGP is unbiased.} \\
	
	Part (a) of the proposition implies that, whenever the agent's perception is not perfect, there exists a mechanism that provides the agent with a positive value of attention. That is, there is always scope to incentivize attention through the design of the mechanism. In turn, such a mechanism must feature some screening by part (b) of the proposition. Indeed, a mechanism that does not discriminate between types does not provide incentives to correctly perceive types. While these results are very intuitive, they serve as sanity checks on our modeling approach.\\
	
	Part (b) of Proposition~\ref{prop:gen Z} implies that screening is necessary to create strictly positive attention incentives. We now turn to the question of how to maximize such incentives. Recall that the set of feasible allocation rules in our problem, $Q$, consists of any non-decreasing $q:[0, 1] \to [0, 1]$. Let $Q^* \subset Q$ denote the set of all allocation rules that maximize the value of attention. In Proposition~\ref{prop:max_characterization} in the appendix, we fully characterize $Q^*$.\footnote{Since \eqref{eqn:vos} is a linear functional on $Q$ it attains a maximum at one of its extreme points (i.e., the set of threshold allocation rules). Moreover, every maximizer is essentially a ``convex-combination" of threshold maximizers. See the discussion around Proposition~\ref{prop:max_characterization} in the appendix for details.} We now state the important qualitative properties that all allocation rules in $Q^*$ satisfy. 

 \begin{prop}\label{prop:coarsescreening}
		Suppose that either $\rho$ is biased or $F \ne F_I$. Then, there exist unique disjoint intervals $\underline{\Pi}\subset [0, 1]$ (with $0 \in \underline{\Pi}$) and $\overline{\Pi} \subset [0,1]$ (with $1 \in \overline{\Pi}$) such that:
		\begin{enumerate}[(a)]
			\item $q(\pi) = 0$ for all $q \in Q^*$ if and only if $\pi \in \underline{\Pi}$; and
			\item $q(\pi) = 1$ for all $q \in Q^*$ if and only if $\pi \in \overline{\Pi}$.
		\end{enumerate} 
	\end{prop}

 Proposition~\ref{prop:coarsescreening} implies that, if $q$ maximizes the value of attention, it must exclude sufficiently low types (i.e., provide them with allocation of 0) and pool sufficiently high types at maximal allocation (i.e., provide them with allocation of 1).\footnote{Note that the intervals $\underline{\Pi}$ and $\overline{\Pi}$ depend on the primitives, including the PGP and $F$. The proof of Proposition~\ref{prop:coarsescreening} displays how to construct these intervals. Depending on the primitives, these intervals may either be closed or not. Specifically, there always exists $\pi_1 \le \pi_2$ such that either $\underline{\Pi} = [0, \pi_1)$ or $\underline{\Pi} = [0, \pi_1]$ and either $\overline{\Pi} = (\pi_2,1]$ or $\overline{\Pi} = [\pi_2, 1]$.} Thus, we identify a novel cost of screening on an intensive margin, in terms of attention incentives. Such coarse screening is optimal as it increases the wedge in the behavior of an attentive and inattentive agent (measured by payoff differences), conditional on requiring the allocation rule to be monotone. Hence, for given $q \in Q$, one can always increase the value of attention either by decreasing $q$ on $\underline{\Pi}$ or increasing $q$ on $\overline{\Pi}$.\footnote{Note that $\underline{\Pi} \cup  \overline{\Pi} $ need not coincide with $  [0,1]$. In this case, there is a continuum of maximizers including both the allocation rules with $q(\pi)=1$ and $q(\pi)=0$ for all  $\pi \notin \underline{\Pi} \cup  \overline{\Pi} $ (see Proposition~\ref{prop:max_characterization}).} The following example illustrates Proposition~\ref{prop:coarsescreening} for a particular PGP.
	
	\begin{example}
		Suppose that $\theta \sim U[0, 1]$ and consider an example of conservatism bias, as described in Section~\ref{sec:examples}, with $\alpha = 1/2$. Specifically, when $\theta$ realizes, an inattentive agent instead perceives type $\frac{1}{2} \theta + \frac{1}{4}$. Proposition~\ref{prop:gen Z} implies that the value of attention can be written as
		$$
		\nu(q) = \int \limits_0^{\frac{1}{4}} q(\pi)[-\pi]d\pi + \int \limits_{\frac{1}{4}}^{\frac{3}{4}} q(\pi)\left[\frac{1}{2} - \pi\right] d\pi + \int \limits_{\frac{3}{4}}^1 q(\pi)[1 - \pi]d\pi.
		$$
        Clearly, if $\pi < 1/4$, decreasing $q(\pi)$ increases the value of attention, while if $\pi > 3/4$, increasing $q(\pi)$ increases the value of attention. As such, all maximizers have $q(\pi) = 0$ if $\pi \in \underline{\Pi} \equiv [0, 1/4)$ and $q(\pi) = 1$ if $\pi \in \overline{\Pi} \equiv (3/4, 1]$. 
	\end{example}
	
	To this point we have described how, for a given PGP, the mechanism impacts on the agent's attention incentives. We now explore how the agent's PGP impacts on the value of attention for a given mechanism. Our objective is to characterize a notion of \emph{accuracy} of perception. Intuitively, the more accurate the agent's PGP is, the more aligned their attentive and inattentive behaviours should be. Moreover, the more aligned these behaviors are, the less incentives the agent has to be attentive and, consequently, the better off the agent should be. Thus, in the spirit of Blackwell's informativeness criterion, we measure accuracy in perception in terms of the agent's welfare (restricted to environments with quasi-linear utility functions
that are bilinear in types and allocations). To this end, let 
$$
 V(\mathcal{M}, \kappa; \rho) \equiv \max \limits\{V_A(\mathcal{M}) - \kappa, V_I(\mathcal{M};\rho)\}
 $$
 denote agent-welfare under mechanism $\mathcal{M}$. We define more accurate perception as follows.

 \begin{definition}\label{defn:accuracy}
     Fix $F$. PGP $\rho$ is more accurate than $\rho^\prime$ if $V(\mathcal{M}, \kappa; \rho) \ge V(\mathcal{M}, \kappa; \rho^\prime)$ for all mechanisms $\mathcal{M}$ and $\kappa \ge 0$.
 \end{definition}

 While our definition of accuracy in perception is hopefully uncontroversial, it is difficult to verify accuracy directly: for two distinct PGPs $\rho$ and $\rho^\prime$ it requires comparing agent-welfare for every mechanism $\mathcal{M}$ and every cost $\kappa$. Fortunately, this condition boils down to a simple statistical property of the agent's perception, which is independent of $\mathcal{M}$ and $\kappa$. For each $x \in [0, 1]$,
 define
\begin{equation}
    S(x; \rho) \equiv \int \limits_x^1 F_I(\pi) d\pi + \int \limits_x^1 [\pi - e_I(\pi)] dF_I(\pi).
\end{equation}
One can interpret $S(\cdot;\rho)$ as a joint measure of both the amount of information contained in (i.e., the first term) and degree of bias exhibited by (i.e., the second term) PGP $\rho$. The following proposition shows that $S(\cdot;\rho)$ is a sufficient statistic for characterizing accuracy.

 \begin{prop}\label{prop:BlackwellNew}
     Fix $F$. The following statements are equivalent:
     \begin{enumerate}[1.]
        \item $\rho$ is more accurate than $\rho^\prime$; and
        \item $S(x; \rho) \le S(x; \rho^\prime)$ for all $x \in [0, 1]$.
     \end{enumerate}
 \end{prop} 

    Due to the simplicity of calculating $S(x; \rho)$, verifying whether the accuracy order is satisfied is now straightforward. The intuition for why $S(\cdot;\rho)$ characterizes accuracy is as follows.
    First, note that the PGP, $\rho$, only impacts on the agent's welfare if they are inattentive; that is, $V_I(\mathcal{M};\rho)$ depends on $\rho$ while $V_A(\mathcal{M})$ does not. Thus, since the cognitive state results from a rational cost-benefit trade-off, the agent can only benefit from having lower opportunity costs from being inattentive (i.e., higher inattentive utility). The statistic $S(\cdot; \rho)$ is a sufficient statistic for ordering inattentive utility because it objectively measures the impacts of an PGP on that utility--i.e., the amount of information a PGP preserves and the degree of bias it introduces. We now utilize Proposition~\ref{prop:BlackwellNew} to establish a number of natural PGPs that are ordered by accuracy.

	\begin{example}[More Information] \label{example:moreinfo}
		Consider the class of unbiased PGPs. Then, by Proposition~\ref{prop:BlackwellNew}, PGP $\rho$ (with $F_I$) is more accurate than $\rho^\prime$ (with $F_I^\prime$) if and only if 
        $$
		\int \limits_x^1 F_I(\pi) d\pi \le \int \limits_x^1 F_I^\prime(\pi)d\pi,
		$$
		for all $x \in [0, 1]$, where equality holds if $x = 0$ since $F_I$ and $F_I^\prime$ imply the same mean. That is, PGP $\rho$ is more accurate than $\rho^\prime$ if and only if $F_I$ majorizes $F_I^\prime$ (see, e.g., \cite{Kleiner2021}). Equivalently, more accurate coincides with more information (when inattentive) in the class of unbiased PGPs.

	\end{example}
	
	\begin{example}[More Bias in Perception] \label{example:morebias}
		   If we fix the amount of true information generated by the agent's PGP (i.e., the distribution of $e_I(\pi)$), then the PGP becomes less accurate as the wedge between true types and their inattentive perceptions increases. For example, consider a case of probability weighting where true type $\theta$ is distributed according to $F$, but is perceived as $\theta^{\alpha}$, with $\alpha \in (0, +\infty)$ (see Section~\ref{sec:examples}). In this case, the PGP is fully determined by $\alpha$ and one can establish that
     $$
     S(x; \alpha) = 1 - x + \int \limits_{x^{1/\alpha}}^1 [x - \pi]   dF(\pi).
     $$
     For any $x \in [0,1]$, this increases in $\alpha$ for $\alpha > 1$ (more pessimism implies lower accuracy) and decreases in $\alpha$ for $\alpha < 1$ (more over-optimism implies lower accuracy). \\

In general, removing bias while holding information generated by the PGP constant increases accuracy. Indeed, for a given distribution of $e_I(\pi)$, the agent is better off if they \emph{could} behave using $e_I(\pi)$ (i.e., using $\pi$ as a signal of their true type) rather than by using $\pi$ as their true type. For example, in the probability weighting example above, $\alpha = 1$ (i.e., perfect perception) is always more accurate than any $\alpha \ne 1$.  A consequence is that, holding information generated by the PGP constant, it is always better to generate perception in an unbiased manner than in a biased one. 
	\end{example}
	
	Example~\ref{example:moreinfo} suggests that, holding bias constant, the agent's PGP is more accurate the more information it reveals. Conversely, Example~\ref{example:morebias} highlights that, holding information constant, the agent's PGP is more accurate the less biased is the use of this information. Thus, if an agent's PGP improves in both dimensions, accuracy also improves.\footnote{An example of a class of PGPs ordered in this sense is that of hype (see Section~\ref{sec:examples}). There, as $h$ increases, the buyer perceives their type to be maximal more often. This implies that (i) bias increases, and (ii) true information decreases as perception $\pi = 1$ contains no information. It is simple to verify that $S(x;h) = \int \limits_x^1 F(\pi)d\pi + h \int \limits_0^x F(\pi) d\pi$, so that accuracy decreases in $h$.} These, however, are not the only situations in which PGPs can be ranked according to accuracy. As the following example shows, we can also order some PGPs when one reveals more information but also exhibits more bias.

	\begin{example}[Biased versus Unbiased PGPs]\label{ex:cvi}
		It is well known that $F$ majorizes $F_I$ if and only if $F_I$ can be generated by an unbiased PGP (see, e.g., \cite{Kleiner2021}). If this is the case, is an unbiased PGP always the most accurate way to generate $F_I$? Consider the following counterexample. Suppose $\theta$ is distributed according to $F$ with mean $\mu$ and $F_I(\pi) = F\left(\frac{\pi - (1 - \alpha)\mu}{\alpha}\right)$ for $\pi \in [0, 1]$, where $\alpha \in (0, 1)$. Then, one can show that $F_I$ is majorized by $F$ and, as such, can be generated by an unbiased perception (denote by $\rho^U$). However, it can also emerge from conservatism bias (see Section~\ref{sec:examples}): the agent perceives type $\theta$ as $\alpha \theta + (1 - \alpha) \mu$ (denote by $\rho^C$). In this case, 
        $$
        S(x; \rho^C) = \int \limits_x^1 F_I(\pi)d\pi + \frac{1 - \alpha}{\alpha} \int \limits_x^1 [\mu - \pi] dF_I(\pi) \quad \text{and} \quad S(x; \rho^U) = \int \limits_x^1 F_I(\pi)d\pi. 
        $$
        Clearly $\int \limits_x^1 [\mu - \pi] dF_I(\pi)$ is bounded above by $\int \limits_0^1 [\mu - \pi] dF_I(\pi)$ and, since $F$ majorizes $F_I$, $\int \limits_0^1 [\mu - \pi] dF_I(\pi)=0$. Thus, generating $F_I$ via conservatism bias is more accurate than doing so in an unbiased manner. Intuitively, while conservatism involves more bias in perception, it also contains more information than its unbiased but observationally-equivalent counterpart. As such, our notion of accuracy allows us to rank some PGPs that conflict in terms of their degrees of bias and informativeness.
	\end{example}
	
	\section{Applications}

		\subsection{Efficient Mechanisms}
	
	In this section, we consider the objective of allocative efficiency in a canonical principal-agent framework. Formally, a principal produces allocation $x \in [0,1]$ at cost $C(x)$, where $C(\cdot)$ is continuous.\footnote{Continuity is sufficient to ensure that there exists a solution to the principal's problem.} 
	The designer's objective is to maximize total surplus. That is, the sum of the agent's utility (including any cognitive costs) net of the costs of provision.
	
	\subsubsection{Exogenous Cognitive States}

	Suppose first that the agent's cognitive state is fixed and known to the designer. If the agent is attentive (i.e., perception $\pi = \theta$) then the best the designer can hope to achieve is
	\begin{equation}\label{eqn:efficient_A}
		W_A^* \equiv \max \limits_{q} \int \limits_0^1 \pi q(\pi)-C(q(\pi)) d F(\pi).
	\end{equation}
	If, instead, the agent is inattentive, then the upper bound on welfare is
	\begin{equation} \label{eqn:efficient_I}
		W_I^* \equiv \max \limits_{q} \int\limits_0^1 \int \limits_0^1 \theta q(\pi) -C(q(\pi)) d\rho(\pi | \theta) dF(\theta) = \int \limits_0^1 e_I(\pi) q(\pi)-C(q(\pi)) d F_I(\pi).
	\end{equation}
	Evidently, $W_I^* \le W_A^*$. We use $W_A^*$ and $W_I^*$ to define notions of efficiency for each of the agent's cognitive states.
	
	\begin{definition}
		Direct mechanism $\mathcal{M} = (q, t)$ is efficient under attention if it achieves $W_A^*$. It is efficient under inattention if it achieves $W_I^*$. 
	\end{definition}
	
	Our objective is to understand whether an agent's PGP allows for efficiency to be achieved for cognitive state $\omega \in \{A, I\}$.
	We first define a mechanism that is efficient under attention, \emph{selling the firm to the agent}:
	
	\begin{definition}
		Direct mechanism $\mathcal{M} = (q, t)$ \textbf{sells the firm to the agent} if $t(\pi) = C(q(\pi))$ and $q(\pi) \in \arg \max \limits_{x \in [0, 1]} \pi x - C(x)$.
	\end{definition}
	
	Clearly, such a mechanism is incentive compatible as the agent is the claimant on welfare. Moreover, this mechanism is efficient under attention as the agent's optimal reporting decision is equivalent to that in \eqref{eqn:efficient_A}. Suppose next that the agent is inattentive and first consider the case
	of an unbiased PGP. Then, as is well known, selling the firm to the agent remains optimal.
	The reason is that an agent operating under an unbiased PGP correctly perceives their \emph{expected} marginal contribution to welfare at every perception (i.e., $e_I(\pi) = \pi)$. Thus, their incentives are perfectly aligned with welfare maximization under such a mechanism, as is clear from \eqref{eqn:efficient_I}. Instead, selling the firm to the agent will, in general, not be optimal if the agent's perception is biased. 
	
	\begin{prop}\label{prop:welfare_exo}
		Suppose the agent's cognitive state is exogenous. Then,
		\begin{enumerate}
			\item \emph{Selling the firm to the agent} achieves efficiency if the agent is attentive or has an unbiased PGP;
			\item If the agent is inattentive and has a biased PGP, there exists a cost function so that \emph{selling the firm to the agent} is not efficient under inattention.
		\end{enumerate}
	\end{prop}
	
    The reason that selling the firm to the agent is sub-optimal when the PGP exhibits bias is that such an agent misperceives how their report/action contributes to welfare, leading to a misalignment between their behavior and the welfare-maximizing objective of the designer. However, if the agent mis-perceives their private information in a systematic way (i.e., $e_I$ is monotone), the designer could, in principle, attune the mechanism to correcting for the agent's misperception. This is what we call \emph{managing the process}.
	\begin{definition}
		For given PGP with $e_I$, a mechanism $\mathcal{M} = (q, t)$ \textbf{manages the process} if $t(\pi) =  q(\pi)\pi -\int_0^{\pi} q(s)ds$ and $q(\pi) \in \arg \max \limits_{x \in [0, 1]} e_I(\pi) x - C(x)$ for all $\pi \in \supp(F_I)$.\footnote{Note that many mechanisms may satisfy this definition, as $supp(F_I)$ does not necessarily coincide with $[0,1]$.}
	\end{definition}
	
	In particular, for each perception the agent holds, the designer offers efficient provision in terms of the true informational content of that perception. Such management of the process has been explored in other contexts, such as efficiency-based justifications for affirmative action \citep{emil2023}.
	\begin{prop}\label{prop:manage}
		Any mechanism that manages the process is efficient under inattention if $e_I$ is non-decreasing. Instead, if $e_I$ is ever decreasing, then there is a cost function $C$ such that no feasible mechanism is efficient under inattention.
	\end{prop}
	Notice the limitation in Proposition~\ref{prop:manage}, demanding $e_I$ to be increasing.
	This requirement holds when the agent's perception of their type is qualitatively correct - higher types hold higher perceptions than lower types - and misperception occurs at a quantitative level.\footnote{Many models of misperception exhibit this property, such as probability weighting and conservatism bias (see Section~\ref{sec:examples}).} A fundamental issue arises, instead, if $e_I(\cdot)$ ever decreases. In this case, the ``managing the process" mechanism may, depending on the specification of $C$, not be feasible in that it requires an allocation rule that is partially decreasing and, as such, is not incentive compatible. Hence, an efficient mechanism may not exist in this case and ironing is required.\footnote{An example of a PGP that generates decreasing $e_I$ is that of hype (see Section~\ref{sec:examples}). In this case, it is simple to see that $e_I(\pi) = \pi $ if $\pi < 1$ and $e(1)=1/2$. Thus, $e_I(\pi)$ decreases at $\pi = 1$. In this case, there exist cost functions for which no mechanism is efficient under inattention; for example, $C(x) = x^2/2$.}

	\subsubsection{Endogenous Cognitive States}
	
	We now consider the case where the agent's cognitive state depends on the mechanism. We establish that, in case the agent's PGP is biased, the problem of over-attention can arise: mechanisms that are efficient under inattention may provide the agent with overly large attention incentives. This leads to two potential sources of inefficiency: either (i) the designer implements a mechanism that is efficient under attention when, from a welfare perspective, having the agent bear attention costs is not desirable, or (ii) the designer distorts a mechanism that manages the process in order to incentivize inattention. Formally, the designer seeks to maximize welfare
	taking into account the agent's attention incentives and costs by solving
	$$\max_{q, \alpha \in [0, 1]}  \alpha \times  \left\{ \int \limits_0^1 e_I(\pi) q(\pi)-C(q(\pi)) d F_I(\pi) \right\} +(1-\alpha) \times \left\{ \int \limits_0^1 \pi q(\pi)-C(q(\pi)) d F(\pi) - \kappa   \right\}, $$
	where $q$ is non-decreasing, $\alpha = 1$ if $\nu(q) > \kappa$, $\alpha = 0$ if $\nu(q) < \kappa$, and $\alpha \in [0, 1]$ if $\nu(q) = \kappa$. \\
	
	Recall that the best the designer can achieve in cognitive state $\omega \in \{A, I\}$ is $W_\omega^*$, where $W_I^* \le W_A^*$. Thus, an upper bound on welfare is $	W^*_\kappa \equiv \max \left\{W_A^* - \kappa, W_I^*\right\},$
 so that incentivizing attention is optimal if and only if $\kappa \le \kappa^* \equiv W_A^* - W_I^*$. 
	
	\begin{definition}
		We say that mechanism $\mathcal{M}$ is efficient for $\kappa$ if it achieves $W_\kappa^*$.
	\end{definition}

	We start again by assuming that the agent's PGP is unbiased and show that selling the firm to the agent is an efficient mechanism for any $\kappa$. This result is reminiscent of \cite{bergemann2002}, who establish that Vickrey-Clarke-Groves (VCG) mechanisms, which provide agents with their marginal contribution to welfare, are efficient in allocation problems with information acquisition.
	\begin{prop} \label{prop:jusso}
		Suppose that the agent's PGP is unbiased. Then, the mechanism that sells the firm to the agent is an efficient mechanism for all $\kappa \ge 0$ and welfare decreases in $\kappa$.
	\end{prop}
	
	The intuition for this result is as follows. If the agent's PGP is unbiased, selling the firm to the agent is efficient under both cognitive states. Moreover, since the agent is the claimant on welfare in this mechanism, the value of attention is equal to $W_A^* - W_I^*$. As such, the agent's attention incentives are perfectly aligned with welfare maximization. \\  
	
	We now consider the case in which the agent's PGP is biased. As already established, if $e_I(\cdot)$ ever decreases then efficiency cannot be achieved for exogenous inattention and, consequently, could not be achieved if $\kappa$ is sufficiently large. If, instead, $e_I(\cdot)$ is non-decreasing, then an efficient mechanism exists under both inattention and attention. Thus, one may wonder whether an efficient mechanism exists for any $\kappa \ge 0$ as long as $e_I(\cdot)$ is increasing. This, however, is not generally true as the issue of \emph{over-attention} may arise. In particular, if an efficient mechanism for inattention generates attention incentives that are too strong, then it may incentivize attention when inattention is actually desired.\\
	
	Recall that, when $e_I$ is non-decreasing, any mechanism that manages the process is both feasible and efficient under inattention. Define
	$$
	\kappa_I \equiv \inf\{\nu(q): \text{$q$ is non-decreasing and manages the process}\}
	$$
	to be the lowest value of attention induced by any mechanism that manages the process. Note that, in general, mechanisms that manage the process induce some screening and, consequently, $\kappa_I > 0$. As such, for $\kappa < \kappa_I$, there does not exist a feasible mechanism that is efficient under inattention that simultaneously induces inattention. The following proposition describes when such over-attention jeopardizes efficiency and characterizes the optimal mechanism. 
	
	\begin{prop}\label{prop:inefficiency}
		Suppose the agent's PGP is biased and $e_I$ is non-decreasing. Then, the optimal mechanism is inefficient for $\kappa$ if and only if $\kappa_I > \kappa^*$ and $\kappa \in (\kappa^*, \kappa_I)$. Moreover, if $\kappa_I > \kappa^*$ then there is $\bar \kappa \in (\kappa^*, \kappa_I) $ such that the optimal mechanism is selling the firm to the agent if $\kappa <\bar \kappa$, while for $\kappa > \bar \kappa$ the optimal mechanism solves 
        $$ \max_q \int_0^1 e_I(\pi)q(\pi)-C(q(\pi))dF_I(\pi) \quad s.t. \quad \nu(q)=\kappa.$$  
	\end{prop}
	Figure~\ref{fig:cases} illustrates Proposition~\ref{prop:inefficiency}.
	In line with the exogenous benchmark, efficiency for $\kappa$ can be achieved if cognitive costs are sufficiently small (i.e., $\kappa \leq \kappa^*$) by selling the firm to the agent, or sufficiently large (i.e., $\kappa \geq \kappa_I$) by managing the process. To see why selling the firm to the agent, which is efficient under attention, is an optimal mechanism if the designer wants the agent to pay attention, note that this mechanism perfectly aligns (rational) attention incentives with welfare; that is, the agent's utility under this mechanism, regardless of their their realized cognitive state, coincides with actual welfare. As such, selling the firm to the agent is, independent of $\kappa$, a feasible mechanism that serves as a lower bound on the designer's objective. Thus, introducing distortions to this mechanism to promote attention is never optimal, as this would lead to a strictly lower level of welfare.  \\

     \begin{figure}[http]
     \centering
         
 
\begin{tikzpicture}
    % Axes

   % Top graph horizontal line with curly brackets above
    \draw[thick,->] (0,6) -- (6,6);
    \draw [decorate,decoration={brace,amplitude=10pt,raise=4pt},yshift=0pt] (0,6) -- (2,6) node [black,midway,yshift=0.6cm,above] {$W_A^*$};
    \draw [decorate,decoration={brace,amplitude=10pt,raise=4pt},yshift=0pt] (2,6) -- (4,6) node [black,midway,yshift=0.6cm,above] {Inefficient};
    \draw [decorate,decoration={brace,amplitude=10pt,raise=4pt},yshift=0pt] (4,6) -- (6,6) node [black,midway,yshift=0.6cm,above] {$W_I^*$};
     \draw[below] (5.5,6) node {$\kappa$};
      \draw[below] (13.5,6) node {$\kappa$};
    \draw[below] (2,6) node {$\kappa^*$};
\draw[below] (4,6) node {$\kappa_I$};
    % Bottom graph horizontal line with curly brackets above
    \draw[thick,->] (8,6) -- (14,6);
    \draw [decorate,decoration={brace,amplitude=10pt,raise=4pt},yshift=0pt] (8,6) -- (11,6) node [black,midway,yshift=0.6cm,above] {$W_A^*$};
    \draw [decorate,decoration={brace,amplitude=10pt,raise=4pt},yshift=0pt] (11,6) -- (14,6) node [black,midway,yshift=0.6cm,above] {$W_I^*$};
 \draw[below] (11,6) node {$\kappa^*$};
\draw[below] (9,6) node {$\kappa_I$};
\draw (9,5.9)--(9,6.1);
\draw (11,5.9)--(11,6.1);
\draw (2,5.9)--(2,6.1);
\draw (4,5.9)--(4,6.1);
\draw (0,5.7)--(0,6.3);
\draw (8,5.7)--(8,6.3);
\draw[below] (3,6) node {$\bar \kappa $};
\draw (3,5.9)--(3,6.1);
\end{tikzpicture}
     \caption{\small{The left panel depicts a situation where $\kappa^*<\kappa_I$. If $\kappa^* < \kappa <\bar \kappa$ the agent bears (inefficiently high) attention costs; instead, if $\kappa_I>\kappa>\bar \kappa$ the optimal mechanism distorts the inattentive efficient allocation to keep the agent inattentive. The right panel depicts the case  $\kappa^*>\kappa_I$. There the mechanism is efficient for any $\kappa$.}}
     \label{fig:cases}
 \end{figure}
	 
	Efficiency, however, can be achieved for every $\kappa$ only if over-attention is not an issue. By this we mean that the attention incentives of mechanisms that manage the process are not overly strong (i.e., $\kappa_I \le \kappa^*$). In this case, any mechanism that manages the process with value of attention below $\kappa$ is efficient for $\kappa \ge \kappa^*$. Instead, when every efficient mechanism for inattention comes with strong attention incentives (i.e., $\kappa_I > \kappa^*$), then a region of inefficiency is unavoidable as depicted in the left panel of Figure~\ref{fig:cases}: for $\kappa \in (\kappa^*, \kappa_I)$ the designer wants the agent to be inattentive but every mechanism that is efficient under inattention induces attention.\footnote{Over-attention (i.e., $\kappa_I > \kappa^*$) is most likely to arise if either (i) managing the process under inattention (almost) achieves attentive first-best (i.e., $W_A^* \approx W_I^*$ so that $\kappa^* \approx 0$), or (ii) mechanisms that are efficient under inattention requires coarse screening as identified in Proposition~\ref{prop:coarsescreening}.} This is because,
    in contrast to selling the firm to the agent, mechanisms that manage the process (and thus are efficient under inattention) do not necessarily align the agent's value of attention with the objective of welfare maximization. In this case, the designer must introduce inefficiency, either in the form of inducing attention (when it is overly costly to do so) or distorting a mechanism that manages the process to decrease attention incentives. The following example illustrates the issue of over-attention in welfare maximization problems. \\

	\begin{example}
		Suppose that $\theta$ is uniformly distributed on $[0, 1]$ and the agent exhibits probability weighting when inattentive: they perceive type $\theta$ to be $\theta^{1/2}$. Moreover, costs are given by $C(x) = \frac{1}{4}x$ for $x \in [0, 1]$. Suppose first that the agent's cognitive state is exogenous. The mechanisms that are efficient for cognitive state $\omega \in \{A, I\}$ are
		$$
		q_A(\pi) \equiv \begin{cases}
			1 & \text{if $\pi \ge 1/4$}\\
			0 & \text{if $\pi < 1/4$}
		\end{cases};
		\quad 
		q_I(\pi) \equiv \begin{cases}
			1 & \text{if $\pi \ge 1/2$}\\
			0 &  \text{if $\pi < 1/2$},
		\end{cases}
		$$
		where $q_A$ sells the firm to the agent and $q_I$ optimally manages the process as $e_I(\pi) = \pi^2$ (which is increasing) and $e_I(1/2) = 1/4$. In this case, $W_A^* = W_I^* = 9/32$, so that exogenous inattention does not imply an efficiency loss.\\
		
		Now, suppose that the cognitive state is endogenous and observe that the value of attention can be written as
		$
		\nu(q) = \int \limits_0^1 q(\pi)[3\pi^2 - 2\pi^3 - \pi] d\pi,
		$
		which is maximized by the allocation rule $q(\pi) = 1$ for $\pi \ge 1/2$ and $q(\pi) = 0$ for $\pi < 1/2$. Consequently, the mechanism that is efficient under inattention also maximizes the value of attention, with solution value $\kappa_I \equiv 1/32$. \\
		
		Since $\kappa^* = W_A^* - W_I^* = 0$, it follows that $\kappa^* < \kappa_I$ so that Proposition~\ref{prop:inefficiency} implies that the optimal mechanism can never be efficient for $\kappa \in (0, 1/32)$. The optimal mechanism for $\kappa < 1/32$ is as follows:
		\begin{itemize}
			\item For $\kappa \in (9/512, 1/32)$, $q^*(\pi) = 1$ if $\pi \ge p$ and $q^*(\pi) = 0$ if $\pi < p$, where $p \in (1/4, 1/2)$ is the unique solution to 
			$$
			\frac{1}{2}\hspace{0.5mm}p^2(1 - p)^2= \kappa,
			$$
			which ensures the value of attention does not exceed $\kappa$ so that the agent is inattentive.
			\item For $\kappa < 9/512$, $q^* = q_A$. The agent is attentive in this case. Note that the mechanism sells the firm to the agent in this case, consistent with Proposition~\ref{prop:inefficiency}.
		\end{itemize}
		Figure~\ref{fig:wel} provides a plot of welfare in the optimal mechanism. Notice that welfare actually \emph{increases} in $\kappa$ on the interval $(9/512, 9/32)$. This stands in contrast to Proposition~\ref{prop:jusso}, which states that welfare can never increase in $\kappa$ when the agent's perception is unbiased.

	\end{example}

    \begin{figure}[http]
		
		\hspace{4cm}
		\begin{tikzpicture}[xscale=110,yscale=10]
	
	\draw[->] (0,0) -- (0.036,0) node[below] { $\kappa$};
	\draw[->] (0,0) -- (0,0.25) node[left] {$\textcolor{blue}{W_A-\kappa},\textcolor{red}{W_I}$};
	\draw[blue, domain=0:9/512,thick] plot (\x , {2*(1/4+1/32 -\x-0.1)-0.18});
%	\draw[red, domain=9/512:1/32,thick] plot (\x , {0.25+0.25*  (1 -( (1 - 4 (2*\x)^(0.5))^(0.5)))-0.5 *  (1 -( (1 - 4 (2*\x)^(0.5))^(0.5)))});
	\draw[domain=9/512:1/32, red, thick] 
	plot ({\x}, 
	{2*(1/4 + 1/4*(1/2*(1 - sqrt(1 - 4*sqrt(2)*sqrt(\x))))^2 - 
		1/2*(1/2*(1 - sqrt(1 - 4*sqrt(2)*sqrt(\x))))^4-0.1)-0.18});
	\draw[red, domain=1/32:0.035,thick] plot (\x , {1/4+1/32-0.1 });
 	\draw[-,thick] (0.25,0.01)--(0.25,-0.01) node[below] {$0.25$};
 		\draw[-,thick] (0.25,0.01)--(0.25,-0.01) node[below] {$0.25$};
 		%	\draw[-,thick] (0.25,0.01)--(0.25,-0.01) node[below] {$0.25$};
 	\draw[-,dashed] (1/32,9/32-0.1)--(1/32,-.000001) node[below]{$ \kappa_I$};
 	 	\draw[-,dashed] (9/512,(1/4+1/32-9/512-0.12)--(9/512,-.000001) node[below]{$ 9/512$};
    \draw[-,dashed] (0, 1/4+1/32-0.1)--(1/32,1/4+1/32-0.1);
\end{tikzpicture} %
		
		\vspace*{-4mm}
		
		\caption{\small{The blue curve graphs the agent's maximized expected utility when attentive. The red curve that when operating under inattention.}}
		\small{\label{fig:wel}}
	\end{figure}
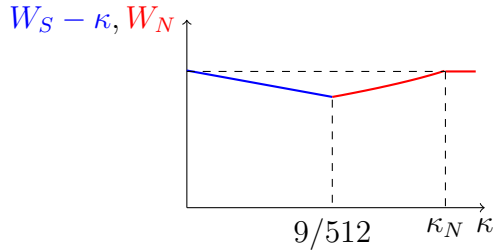 
\subsection{Profit-Maximization through Screening}

In the previous application, we examined aligned preferences between the designer and the agent. Now, we investigate misaligned preferences, where a profit-maximizing firm optimally screens an agent with endogenous perception. Using the standard framework of \cite{Mussa1978}, the agent values provision of quality $x \in [0, 1]$ at $\theta x$, where we assume $\theta \sim U[0, 1]$ for simplicity. The firm offers a menu of products $(q(\theta), t(\theta))_{\theta \in [0, 1]}$. Provision to type $\theta$ of $q(\theta) = x$ costs the firm $C(x) = x^2/2$.\\

A prominent result from information economics is that the final distribution of beliefs that an information structures generates is a sufficient statistic for it. The exact (unbiased) process of how these perceptions are generated (i.e., the PGP) has no impact on the design of the optimal mechanism. In this application we show that, for endogenous cognitive states, if the PGP is biased the final distribution of perceptions is not a sufficient statistic and the exact PGP impacts on the optimal mechanism. To do so, we focus on a situation where the inattentive agent's perception distribution is uniform on $[1/4, 3/4]$. We consider two different perceptions that generate this distribution: an unbiased perception, $\rho_U$,\footnote{Note that $\rho_U$ is not uniquely pinned down. There are many ways to generate the desired perception distribution in an unbiased way. However, as long as they are unbiased, which one to use has no impact on the results of this section.} and conservatism bias, $\rho_C$, where the agent perceives type $\theta$ as $\frac{1}{2} \theta + \frac{1}{4}$. Interestingly, the biased perception is more accurate than the unbiased one and is thus preferred by the agent in any mechanism (see Example~\ref{ex:cvi}). As we will see, the optimal mechanism takes very different forms depending on whether the PGP is $\rho_U$ or $\rho_C$.

\subsubsection{Exogenous Cognitive States}

If the agent's cognitive state is exogenous, only the final distribution of their perception matters. For exogenous attention, the agent's perception is uniformly distributed on $[0, 1]$. The firm chooses $q$ to solve
$ \max \limits_{q} \int \limits_0^1 q(\theta) \left[2 \theta - 1\right] - \frac{q(\theta)^2}{2} d \theta,$ which has the solution $q_A(\theta) = \max \left\{0, 2 \theta - 1\right\}$ for all $\theta \in [0, 1]$. This rule earns the firm profits of $1/12$.\\  

If the agent is exogenously inattentive, their perception is uniformly distributed on $[1/4, 3/4]$ regardless of their PGP. Consequently, the optimal mechanism is the same for both $\rho_C$ and $\rho_U$: the firm chooses $q$ to solve $\max \limits_{q} \int \limits_0^1 \left[q(\pi) \left(2 \pi - \frac{3}{4} \right) - \frac{q(\pi)^2}{2}\right]2 d\pi,$
which has the solution $q_I(x) = \max \left\{0, 2x - \frac{3}{4} \right\}$ for all $x \in [1/4, 3/4]$.\footnote{Note that $q_I(x)$ for $x < 1/4$ must be 0, but for $x > 3/4$, it can be any increasing function on $[3/4, 1]$.} This earns the firm profits of $9/128$.\footnote{Note that if the agent's PGP is unbiased, this grants welfare of $9/256$. If the agent's PGP exhibits conservatism bias, the agent's welfare equals $9/128$. Note that agent welfare is higher with $\rho_C$ than with $\rho_U$, consistent with $\rho_C$ being more accurate than $\rho_U$ and Proposition~\ref{prop:BlackwellNew}.}

\subsubsection{Endogenous Cognitive States}

We now consider endogenous cognitive states. In the exogenous cognitive state benchmark, profits are higher if the agent is attentive compared to inattentive. This implies that the designer has incentives to induce attention. To do so, the designer must increase the value of attention, i.e., the wedge between the agent's benefits of attention ($V_A(\mathcal{M})$) and inattention ($V_I(\mathcal{M})$). The designer can increase this value by either increasing $V_A$ relative to $V_I$ (offering a 'carrot') or decreasing $V_I$ relative to $V_A$ (using a 'stick'). Whether the designer uses a stick or a carrot depends crucially on the agent's PGP. We begin with the unbiased PGP, $\rho_U$.\\

If $\kappa$ is either sufficiently small or large, the designer treats the agent as if their cognitive state was exogenous. The interesting case arises when $\kappa$ is intermediate (i.e., there exist two thresholds $\underline{\kappa}_U<\bar{\kappa}_U$ and $\kappa \in (\underline{\kappa}_U,\bar{\kappa}_U)$). Here, the designer wants the agent to be attentive, but the attention-optimal allocation, $q_A$, does not provide sufficient incentives for attention. Consequently, the designer must introduce distortions. Figure~\ref{fig:Info} illustrates the optimal way for the designer to introduce these distortions for $\rho_U$.

\begin{figure}[h]
	\centering
	\begin{subfigure}[h]{.5\textwidth}
		 \begin{tikzpicture}[xscale=4.5,yscale=5]
	% Left panel (v(theta))
	\draw[->] (0,0) -- (1.2,0) node[below] { $\pi$};
	\draw[->] (0,-0.3) -- (0,0.3) node[above] {$v(\pi)$};
	\draw[blue, domain=0:0.25,thick] plot (\x, {-\x});
	\draw[blue, domain=0.25:0.75,thick] plot (\x, { \x-0.5});
	\draw[blue, domain=0.75:1,thick] plot (\x, {1-\x});
	\draw[-,thick] (0.25,0.01)--(0.25,-0.01) node[below] {$0.25$};
	\draw[-,thick] (0.5,0.01)--(0.5,-0.01) node[below] {$0.5$}; 
	\draw[-,thick] (1,0.01)--(1,-0.01) node[below] {$1$};
	\draw[-,thick] (-0.01,-0.25)--(0.01,-0.25) node[left] {$ -0.25$};
	\draw[-,thick] (-0.01,0.25)--(0.01,0.25) node[left] {$ 0.25$};
	\draw[-,thick] (0.75,0.01)--(0.75,-0.01) node[below] {$0.75$}; 
	%\node[blue] at (0.3,0.7) {$v(\theta)$};
\end{tikzpicture} %
	\end{subfigure}%
	\begin{subfigure}[h]{.5\textwidth}
		 
\begin{tikzpicture}[xscale=4.5,yscale=3]
	
	\draw[->] (0,0) -- (1.2,0) node[below] { $\pi$};
	\draw[->] (0,0) -- (0,1.2) node[left] {\textcolor{blue}{$q_A$}, \textcolor{red}{$q_U^*$}};
	\draw[purple, domain=0:1/2,thick] plot (\x , {0});
%	\draw[blue, domain=0.25:0.75] plot (\x , {0.5});
	\draw[blue, domain=1/2:1,thick] plot (\x , {2*\x-1});
		\draw[red, domain=0.5:3/4,thick] plot (\x , {(2+1)*\x-(1+0.5)});
			\draw[red, domain=3/4:1,thick] plot (\x , {(2-1)*\x-(0)});
%	\draw[-,thick] (0.25,0.01)--(0.25,-0.01) node[below] {$0.25$};
 	\draw[-,thick] (3/4,0.01)--(3/4,-0.01) node[below] {$3/4$}; 
 	\draw[-,thick] (1/2,0.01)--(1/2,-0.01) node[below] {$1/2$}; 
	\draw[-,thick] (1,0.01)--(1,-0.01) node[below] {$1$};
%	\draw[-,thick] (-0.01,3/4)--(0.01,3/4) node[left] {$ 3/4$};
	\draw[-,thick] (-0.01,1)--(0.01,1) node[left] {$ 1$};
	\draw[-,thick] (-0.01,0 )--(0.01,0 ) node[left] {$ 0 $};
\end{tikzpicture}  %
	\end{subfigure}%
	\vspace*{-5mm}
	\caption{The left figure depicts the (marginal contribution to) the value of attention under PGP $\rho_U$. The right figure shows an (optimal) allocation rule in this case for some $\kappa$.}
	\small{\label{fig:Info}}
\end{figure}

As Figure~\ref{fig:Info} shows, the marginal contribution to the value of attention is positive for all perceptions above $1/2$, so the designer increases provision uniformly to the agent relative to $q_A$. This policy change works as a carrot: In the benchmark of exogenous attention ($\kappa=0$), the designer rationed the allocation, compared to what would maximize welfare, to capture most of that welfare. To incentivize attention, the designer introduces less rationing, receiving a smaller share from the increased welfare. As the residual claimant on welfare, the agent benefits from this policy change.\\

Next, consider the biased PGP, $\rho_C$, and assume again that $\kappa$ is intermediate (i.e., there exist two thresholds $\underline{\kappa}_C<\bar{\kappa}_C$ and $\kappa \in (\underline{\kappa}_C,\bar{\kappa}_C)$) . Figure~\ref{fig:Con} illustrates that the marginal contribution of perceptions between $1/2$ and $3/4$ is now negative, so the designer decreases provision to these perceptions relative to $q_A$.

\begin{figure}[h]
	\centering
	\begin{subfigure}[h]{.5\textwidth}
		 \begin{tikzpicture}[xscale=4.5,yscale=5]
	% Left panel (v(theta))
	\draw[->] (0,0) -- (1.2,0) node[below] { $\pi$};
 	\draw[->] (0,-0.3) -- (0,0.3) node[above] {$v(\pi)$};
	\draw[blue, domain=0:0.25,thick] plot (\x, {-\x});
	\draw[blue, domain=0.25:0.75,thick] plot (\x, {0.5-\x});
	\draw[blue, domain=0.75:1,thick] plot (\x, {1-\x});
	\draw[-,thick] (0.25,0.01)--(0.25,-0.01) node[below] {$0.25$};
	\draw[-,thick] (0.5,0.01)--(0.5,-0.01) node[below] {$0.5$}; 
	\draw[-,thick] (1,0.01)--(1,-0.01) node[below] {$1$};
	\draw[-,thick] (-0.01,-0.25)--(0.01,-0.25) node[left] {$ -0.25$};
	\draw[-,thick] (-0.01,0.25)--(0.01,0.25) node[left] {$ 0.25$};
	\draw[-,thick] (0.75,0.01)--(0.75,-0.01) node[below] {$0.75$}; 
	%\node[blue] at (0.3,0.7) {$v(\theta)$};
\end{tikzpicture} %
	\end{subfigure}%
	\begin{subfigure}[h]{.5\textwidth}
		\begin{tikzpicture}[xscale=4.5,yscale=3]
	
	\draw[->] (0,0) -- (1.2,0) node[below] { $\pi$};
	\draw[->] (0,0) -- (0,1.2) node[left] {\textcolor{blue}{$q_A$}, \textcolor{red}{$ q_C^*$}};;
	\draw[purple, domain=0:1/2,thick] plot (\x , {0});
	%	\draw[blue, domain=0.25:0.75] plot (\x , {0.5});
	\draw[blue, domain=1/2:1,thick] plot (\x , {2*\x-1});
	\draw[red, domain=0.5:3/4,thick] plot (\x , {(2-1)*\x-(1-0.5)});
	\draw[red, domain=3/4:1,thick] plot (\x , {(2-1)*\x-(0)});
	%	\draw[-,thick] (0.25,0.01)--(0.25,-0.01) node[below] {$0.25$};
	\draw[-,thick] (3/4,0.01)--(3/4,-0.01) node[below] {$3/4$}; 
	\draw[-,thick] (1/2,0.01)--(1/2,-0.01) node[below] {$1/2$}; 
	\draw[-,thick] (1,0.01)--(1,-0.01) node[below] {$1$};
%	\draw[-,thick] (-0.01,3/4)--(0.01,3/4) node[left] {$ 3/4$};
	\draw[-,thick] (-0.01,1)--(0.01,1) node[left] {$ 1$};
	\draw[-,thick] (-0.01,0 )--(0.01,0 ) node[left] {$ 0 $};
\end{tikzpicture}  %
	\end{subfigure}%
	\vspace*{-4mm}
	\caption{\small{The left figure depicts the (marginal contribution to) the value of attention under PGP $\rho_C$. The right figure shows an (optimal) allocation rule in this case for some $\kappa$.}}
	\label{fig:Con}
\end{figure}

The change in the designer's allocation policy works as a 'stick'. While the average (over all types) allocation stayed constant, rationing for types between $3/4$ and $1/2$ increased, making them worse off. Although the allocation increased for types at the top, they do not benefit from this new policy as it is less profitable to imitate lower types, and thus the designer charges them a higher price.

Figure~\ref{fig:du} and Proposition~\ref{prop:carrotstick} summarize the above discussion.

\begin{prop}\label{prop:carrotstick}
	For intermediate costs, to incentivize attention, the designer uses:
	\begin{enumerate}[(a)]
		\item \textbf{a carrot if the PGP is unbiased:} $V_A$ and $V_I$ increase in $\kappa$ on $(\underline{\kappa}_U,\overline{\kappa}_U)$; and
		\item \textbf{a stick if the PGP is biased:} $V_A$ and $V_I$ decrease in $\kappa$ on $(\underline{\kappa}_C,\overline{\kappa}_C)$
	\end{enumerate}
\end{prop}

\begin{figure}[h]
	\centering
	\begin{subfigure}[h]{.5\textwidth}
		\begin{tikzpicture}[xscale=100,yscale=75]
	
	\draw[->] (0,0) -- (1/16,0) node[below] { $\kappa$};
	\draw[->] (0,0) -- (0,1/16) node[left] {$\textcolor{blue}{V_A},\textcolor{red}{V_I}$};
	\draw[blue, domain=0:1/32,thick] plot (\x , {1/24});
		\draw[red, domain=0:1/32,thick] plot (\x , {1/96});
	%	\draw[blue, domain=0.25:0.75] plot (\x , {0.5});
	\draw[-,dashed] (1/32,0)--(1/32,1/24);
		\draw[-,dashed] (1/32,0.000001)--(1/32,-.000001) node[below]{$\underline{\kappa}_U$};
			\draw[-,dashed] (0.0477,0.000001)--(0.0477,-.000001) node[below]{$\overline{\kappa}_U$};
				\draw[-,dashed] (0.0477,0 )--(0.0477,0.067); 
				\draw[<->,dashed] (0.0394,0.0539) -- (0.0394,0.0145);
				\node [right]at (0.0394,0.03) {$\kappa$};
	\draw[blue, domain=1/32:sqrt(2.5)/96+1/32,thick] plot (\x , {3/2*\x-1/192});
	\draw[red, domain=1/32:sqrt(2.5)/96+1/32,thick] plot (\x , {1/2*\x-1/192});
%	\draw[red, domain=0.5:3/4,thick] plot (\x , {(2-1)*\x-(1-0.5)});
%	\draw[red, domain=3/4:1,thick] plot (\x , {(2-1)*\x-(0)});
	%	\draw[-,thick] (0.25,0.01)--(0.25,-0.01) node[below] {$0.25$};
%	\draw[-,thick] (3/4,0.01)--(3/4,-0.01) node[below] {$3/4$}; 
%	\draw[-,thick] (1/2,0.01)--(1/2,-0.01) node[below] {$1/2$}; 
%	\draw[-,thick] (1,0.01)--(1,-0.01) node[below] {$1$};
	%	\draw[-,thick] (-0.01,3/4)--(0.01,3/4) node[left] {$ 3/4$};
%	\draw[-,thick] (-0.01,1)--(0.01,1) node[left] {$ 1$};
%	\draw[-,thick] (-0.01,0 )--(0.01,0 ) node[left] {$ 0 $};
\end{tikzpicture} %
	\end{subfigure}%
	\begin{subfigure}[h]{.5\textwidth}
		\begin{tikzpicture}[xscale=200,yscale=75]
	
	\draw[->] (0,0) -- (1/32,0) node[below] { $\kappa$};
	\draw[->] (0,0) -- (0,1/16) node[left] {$\textcolor{blue}{V_A},\textcolor{red}{V_I}$};
	\draw[blue, domain=0:1/96,thick] plot (\x , {1/24});
	\draw[red, domain=0:1/96,thick] plot (\x , {1/32});
	%	\draw[blue, domain=0.25:0.75] plot (\x , {0.5});
	\draw[-,dashed] (1/96,0)--(1/96,1/24);
	\draw[-,dashed] (1/96,0.000001)--(1/96,-.000001) node[below]{$\underline{\kappa}_C$};
	\draw[-,dashed] (0.0269,0.000001)--(0.0269,-.000001) node[below]{$\overline{\kappa}_C$};
	\draw[-,dashed] (0.0269,0 )--(0.0269,0.0334); 
		\draw[<->,dashed] (0.0187,0.0188)--(0.0187,0.0375); 
		\node [right] at (0.0188,0.028) {$\kappa$};
	\draw[blue, domain=1/96:sqrt(2.5)/96+1/96,thick] plot (\x , {3/64-\x/2});
	\draw[red, domain=1/96:sqrt(2.5)/96+1/96,thick] plot (\x , {3/64-(3/2)*\x});
%	\node[blue] at (0.0269,0.0334) {$U_S$};
%	\node[red] at (0.0269,0.0065) {$U_S$};
	%	\draw[red, domain=0.5:3/4,thick] plot (\x , {(2-1)*\x-(1-0.5)});
	%	\draw[red, domain=3/4:1,thick] plot (\x , {(2-1)*\x-(0)});
	%	\draw[-,thick] (0.25,0.01)--(0.25,-0.01) node[below] {$0.25$};
	%	\draw[-,thick] (3/4,0.01)--(3/4,-0.01) node[below] {$3/4$}; 
	%	\draw[-,thick] (1/2,0.01)--(1/2,-0.01) node[below] {$1/2$}; 
	%	\draw[-,thick] (1,0.01)--(1,-0.01) node[below] {$1$};
	%	\draw[-,thick] (-0.01,3/4)--(0.01,3/4) node[left] {$ 3/4$};
	%	\draw[-,thick] (-0.01,1)--(0.01,1) node[left] {$ 1$};
	%	\draw[-,thick] (-0.01,0 )--(0.01,0 ) node[left] {$ 0 $};
\end{tikzpicture} %
	\end{subfigure}%
	\vspace*{-4mm}
	\caption{\small{The left figure depicts the agent's utility levels depending on their cognitive state if the PGP is unbiased, i.e., for $\rho_U$. The right figure plots the same quantities for a biased PGP, i.e., $\rho_C$.}}
	\small{\label{fig:du}}
\end{figure}
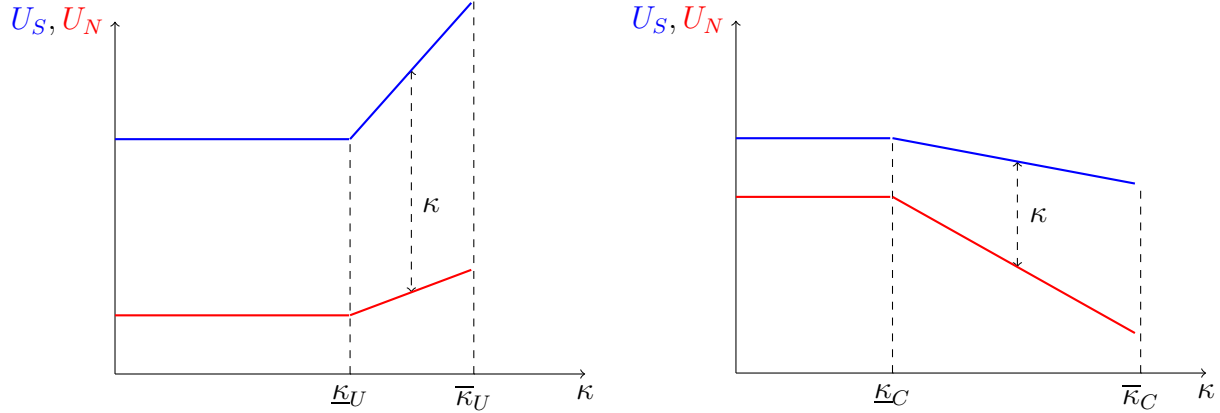

\subsection{`Hype' and Consumer Protection}

In this application, we explore the role of `hype' in product markets. We model hype as a situation where a buyer sometimes purchases a good independently of their valuation for it. We highlight that, although when exogenously biased the seller exploits the buyer, an endogenously attentive agent might benefit from succumbing to hype. This suggests that there is no normative foundation for the principle of always de-biasing consumers. More starkly, we show that, whenever the seller prefers to de-bias the buyer, the buyer actually prefers to succumb to hype more often. Consequently, our model suggests that one should be wary of sellers who enact policies to de-bias consumers: they may do so for purely to maximize profits rather than for benevolent reasons. \\ 

We model hype in product markets as follows. There is a seller and a buyer. The seller chooses a price $p \in [0, 1]$ at which to sell a product to the buyer.\footnote{This restriction to take-it-or-leave-it offers can be shown to be without loss of generality within the space of all feasible mechanisms.} The buyer's value for the product, $\theta$, is uniformly distributed on $[0, 1]$ and is their private information. We utilize the model of hype described in Section~\ref{sec:examples}. Specifically, if the consumer is inattentive when evaluating their type, they succumb to hype with probability $h \in [0, 1]$. If the consumer succumbs to hype, they perceive their type to be maximal (i.e., $\pi = 1$) and will always purchase the product regardless of their true type. Conversely, with $1-h$ probability, the consumer avoids succumbing to hype and correctly perceives their type (i.e., $\pi = \theta$). We assume $h$ increases with the amount of hype surrounding the product. Consequently, we call $h$ the \emph{degree of hype}. If attentive when evaluating their type, the consumer always avoids succumbing to hype.\footnote{This model is similar in spirit to \cite{young2022}, who focuses on whether competition increases or decreases efficiency in a setting where the consumer (deterministically) over-values the product when inattentive. In contrast to that paper, our buyer has private information and stochastically misperceives their valuation when inattentive. Moreover, we focus on the extent to which market participants \emph{want} the buyer to succumb to hype, rather than the impact of market structure, in the form of the degree of competition, on market efficiency.}

\subsubsection{Exogenous Cognitive States}\label{sec:hype_exo}

As a benchmark, we first suppose that the buyer's cognitive state is exogenously determined and known to the seller. Clearly, if the buyer is exogenously attentive, the seller sets the monopoly price of $1/2$, earning the seller revenue of $1/4$ and the buyer welfare of $1/8$. If the buyer is exogenously inattentive, the seller chooses $p$ to solve $\max \limits_{p \in [0, 1]} \left[h + (1 - h)(1 - p)\right]p, $
where $h + (1 - h)(1 - p)$ denotes the probability of trade given the degree of hype $h$ and price $p$. Revenues are maximized at $p_I(h) = \min\left\{1, (2(1 - h))^{-1} \right\}$. Thus, the seller's optimal price increases in $h$, and for $h \ge 1/2$, the seller chooses the maximal price of $1$. It is straightforward to establish that revenues increase in $h$ while $V_I(h)$ decreases in $h$. Hence, the seller prefers to generate as much hype as possible, while the buyer prefers to succumb to hype as rarely as possible. Indeed, the buyer strictly prefers never to succumb to hype (i.e., $h = 0$), which precisely implements the attention-optimal benchmark.

\subsubsection{Endogenous Cognitive States}\label{sec:hype_endo}

We now consider the situation where the buyer's cognitive state endogenously depends on the seller's pricing mechanism. 
Figure~\ref{fig:hype} illustrates the equilibrium depending on $\kappa$ and $h$. For fixed $h$, if $\kappa$ is sufficiently large (above the black dotted curve in Figure~\ref{fig:hype}), the seller charges the exogenous inattentive price and the agent is inattentive. In contrast, if $\kappa$ is sufficiently small (i.e., below the blue solid curve in Figure~\ref{fig:hype}) the seller charges the exogenous attentive price and the agent bears the attention costs. Yet, if $\kappa$ is intermediate (i.e., between the black dotted and the blue solid curve in Figure~\ref{fig:hype}), the seller has an incentive to decrease the price to ensure the buyer is inattentive and succumbs to hype. Indeed, for a given price, $p \in [0, 1]$, the agent purchases the good if and only if $\pi \ge p$. The value of attention takes the form
 $ h (\frac{1}{2} - \int \limits_p^1 x dx) = \frac{1}{2} h p^2,$
which increases both in $h$ and $p$, making the degree of hype and prices complements for creating attention incentives. Thus, high prices may incentivize attention, lowering sales and resultant revenue. As long as $\kappa$ falls above the blue solid curve, the impact on profits of a price decrease is outweighed by higher trade and, thus, keeping the buyer inattentive is optimal from the seller's perspective.\\

\begin{figure}[h]
	\centering
	\begin{subfigure}[h]{.5\textwidth}
		\begin{tikzpicture}
	% Axes
	\draw[->] (0,0) -- (5,0) node[right] {$h$};
	\draw[->] (0,0) -- (0,5) node[above] {$\kappa$};
	\draw[dashed,thick] (4,0)--(4,4) ;
	\draw[dashed,thick] (0,4)--(4,4) ;
	\node at (0,4) [left] {\tiny{$1/2$}};
	\node at (4,0) [below] {\tiny{$1$}};
	% Convex increasing then linear curve
	\draw[thick,dotted] plot [smooth,tension=1] coordinates {(0,0) (1,1) (2,3)} -- (4,4);
	\draw[-,thick] (2,-0.1)--(2,0.1) ;
	\node at (2,0) [below] {\tiny{$1/2$}};
	% Linear curve
	\draw[thick,red, dashed] (0,0) -- (4,2.5);
		\draw[thick,red,dashed] (4,2.5) -- (4,4);
	\draw[thick,blue] (0,1.75) -- (0,4); 
     \node at (0.4,2.5) [left,red] {\tiny{s}};
         \node at (0.9,2.2) [right,blue] {\tiny{b}};
	 \draw[,thick,red, ->] (0.4,2.5) -- (0.9,2.5);
	 	 \draw[,thick,blue, <-] (0.4,2.2) -- (0.9,2.2);
	 	 \node at (2.5,2.5) [left,red] {\tiny{s}};
         \node at (2.5,2.2) [left,blue] {\tiny{b}};
	 	  \draw[,thick,red, ->] (2.5,2.5) -- (3,2.5);
	 	 \draw[,thick,blue, ->] (2.5,2.2) -- (3,2.2);
	 	% \draw [thin,dashed, blue] (0,1.75) -- (4,1.75);

          \node at (3.2,1.6) [right,red] {\tiny{s}};
         \node at (2.7,1.3) [left,blue] {\tiny{b}};
	 	 \draw[,thick,red, <-] (2.7,1.6) -- (3.2,1.6);
	 	 \draw[,thick,blue, ->] (2.7,1.3) -- (3.2,1.3);
	 \node at (0,1.75) [left] {\tiny{$9/128$}};
	 	 \node at (4,0.75) [right,blue] {$h_b^*$};
	 	 	 \node at (4,2.5) [right,red] {$h_s^*$};
	% Convex increasing with steep increase at the end
		\draw[thick,blue] plot [smooth,tension=0.5] coordinates {(0,0) (2,0.9) (3.5,1.2) (4,1.25)};
    \draw[thick,blue] (4,1.25)--(4,1.75);
	\draw[thick,-] (-0.1,1.75) -- (0.1,1.75);
    	\draw[thick,-] (-0.1,1.25) -- (0.1,1.25);
        \node at (0,1.25) [left] {\tiny{$1/32$}};
        \draw[thick,-] (-0.1,3) -- (0.1,3);
        \node at (0,3) [left] {\tiny{$1/4$}};
        \draw[thick,-] (-0.1,2.5) -- (0.1,2.5);
        \node at (0,2.5) [left] {\tiny{$1/8$}};
\end{tikzpicture}%
	\end{subfigure}%
	\vspace*{-4mm}
	\caption{\small{The blue solid curve depicts the buyer optimal level of hype, while the red dashed one graphs that of the seller for a given cognitive cost, $\kappa$. Arrows indicate in which direction of $h$ the seller's revenue (red) and the buyer's utility (blue) increase. }}
	\label{fig:hype}
\end{figure} 
Figure~\ref{fig:hype} also depicts in red (dashed) the level of hype the seller prefers (\textcolor{red}{$h_s^*$}) and in blue (solid) that of the buyer (\textcolor{blue}{$h_b^*$}).
  Contrary to the exogenous benchmark, for sufficiently small $\kappa$, the buyer may prefer a strictly positive degree of hype (i.e., $h_b^* > 0$), while the seller may prefer less than maximal hype (i.e., $h_s^* < 1$). Moreover, whenever $h_b^* > 0$, it holds that $h_b^* > h_s^*$. Hence, whenever the buyer wants to be biased, they prefer to be even \emph{more} biased than the seller would prefer them to be. This has important policy implications. First, the fact that $h_b^*$ can be strictly positive suggests there is no normative underpinning for the paternalistic principle that de-biasing consumers (for example, through nudging) is always to their benefit. Second, sellers that explicitly attempt to de-bias their consumers by lowering the impact of hype (e.g., by offering a generous returns policy) may not have the buyer's best interests in mind. Our model suggests, instead, they may do so for purely to maximize profits.\\

This somewhat counter-intuitive finding stems from a novel commitment-flexibility trade-off generated by endogenous cognitive states. Typically, private information generates inefficiencies in trade. In this light, succumbing more to hype (i.e., higher $h$) defaults the buyer to purchase the good more often, thus mitigating some of these inefficiencies. Their cognitive cost, $\kappa$, then serves as the cost accrued to renege on the commitment to purchase the good. Lower commitment costs then afford flexibility (in the form of a credible threat to be attentive and not always purchase the good) if prices are not sufficiently low. An intermediate level of commitment costs balances these forces optimally, so that the buyer can share in the larger gains from trade that their commitment to buy affords. Indeed, the optimal mix of commitment and flexibility from the buyer's perspective is $h = 1$ and $\kappa = 1/32$, so that they have maximal commitment to buy and sufficient flexibility to ensure that they share equally in the gains from trade with the seller (both receive $1/4$ in this case). For precisely this reason, the seller's and the buyer's preferences regarding hype are misaligned whenever the buyer prefers a strictly positive level of hype. Specifically, a buyer that succumbs to hype more often not only increases welfare but also increases their own share of that welfare. Indeed, since the price and the degree of hype enter the value of attention as complements, the seller reacts to an increased degree of hype by decreasing the sales price (as long as this still earns higher profits than incentivizing attention) and, consequently, increasing consumer surplus. The next proposition formally summarizes these insights.

\begin{prop}\label{prop:optimal_hype}
    Comparing buyer-optimal hype, $h_b^*$, to seller-optimal hype, $h_s^*$, we have that
    \begin{enumerate}[(a)]
    \item $h_s^* > h_b^* = 0$ if $\kappa > 9/128$; and
    \item $h_b^* > h_s^* > 0$ if $\kappa < 9/128$. 
    \end{enumerate}
  Moreover, when $h_b^* > h_s^*$, the buyer's utility is increasing and the seller's profits are decreasing over $h \in (h_s^*, h_b^*)$.
\end{prop}

\section{Conclusion}

We have provided a portable framework of endogenous perception of private information within the context of single-agent mechanism design problems. Depending on cognitive effort at the evaluation stage, the agent is either attentive (i.e., perceives their type correctly) or inattentive (i.e., perceives their type with potential bias). We characterized a sufficient statistic for the agent's attention incentives - the value of attention - and explored how this value varies with both the mechanism and the process via which inattentive perception is generated. In particular, we highlighted a novel cost of screening: coarse mechanisms that screen only on an extensive margin maximize attention incentives. Moreover, we used the value of attention to define a notion of accuracy of perception, which allows us to order some perceptions that are seemingly not comparable due to arising from distinct processes. Our framework flexibly incorporates many notions of misperception of private information that one may deem relevant to investigate. In applications, we showcased the impact of allowing for both biased and endogenous perception for the design of optimal mechanisms. In particular, optimal mechanisms are not just shaped by potential biases but also shape the extent to which these biases are present.\\

While this work serves as a useful benchmark for endogenizing the (mis)perception of private information, there are numerous avenues for future research that warrant exploration. First, for the purpose of tractability, we have made the following simplifying assumptions: (i) quasi-linear utility that is, in addition, bilinear in allocation and types, and (ii) binary cognitive states with an exogenous inattentive PGP. It would be fruitful to study whether results generalize or how predictions change with more general specifications (i.e., more general direct utility functions, or richer cognitive hierarchies where PGPs may themselves be endogenous). Second, we have mostly focused on how the design of mechanisms impact on cognitive state determination for a given (known) inattentive PGP. One could, however, also envision using our results in Section~\ref{sec:VOS} regarding how perception formation responds to incentives in order to identify the nature or shape of individuals' biases. While a rigorous treatment of identification of bias falls outside the scope of this paper, we see it as an important avenue for future research. Finally, it would be natural to expand this idea to a multi-agent framework. Such an extension introduces novel complications, as the processes via which the perceptions of \emph{all} agents are generated are relevant to any individual agent. Consequently, a  concept of \emph{cognitive equilibrium} would need to be introduced to predict behavior in such contexts.

\bibliographystyle{jpe}
\bibliography{literature}

\newpage

\appendix 
 
\section{Appendix}

\textbf{Proof of Proposition~\ref{prop:gen Z}}

\begin{proof}
	Using definitions, it is straightforward to observe that the value of attention can be written as
	\begin{eqnarray}\label{VOS0}
		\nu(\mathcal{M})=   \int_0^1  U(\theta|\theta) dF(\theta)- \int_0^1 \int_0^1 U(\pi|\theta)    d \rho(\pi|\theta) dF(\theta).
	\end{eqnarray} 
	Note that that $U(\theta|\theta)=\int_0^{\theta} q(s)ds+U(0|0)$ (by incentive compatibility). Much similar,
	$ U(\pi|\theta)= U(\pi|\pi)-\pi q(\pi)+\theta q(\pi)= \int^{\pi}_0 q(s)ds + U(0|0) +(\theta-\pi)q(\pi)$. Substituting into \eqref{VOS0}, using the definition of $F_I$, and applying partial integration, this term becomes
	$$  \nu(\mathcal{M})=   \int_0^1  q(s) (1-F(s))ds - \int_0^1 q(s) (1-F_I(s)) ds - \int_0^1 \int_0^1 q(\pi)(\theta-\pi) d\rho(\pi|\theta)dF(\theta)$$
	$$ = \int_0^1 q(s)(F_I(s)-F(s)) ds- \int_0^1 \int_0^1 q(\pi)(\theta-\pi) d\rho(\pi|\theta)dF(\theta).$$
	
	Finally, define $\rho(x|\pi)\equiv Pr(\theta\leq x|\pi)$ and $\mu(\theta,\pi)\equiv \int_0^1 \int_0^1 d\rho(\pi|\theta)dF(\theta)$. Note that, by Bayes' rule, $\mu(\theta,\pi)\equiv \int_0^1 \int_0^1 d\rho(\theta|\pi)dF_I(\pi) $. Then,   
	\begin{align*}
		\int_0^1 \int_0^1 q(\pi)\theta d\rho(\pi|\theta)dF(\theta) &=
		\int_{(\theta,\pi) \in [0,1]^2 } q(\pi)\theta d\mu(\theta,\pi)  \\ &= \int_0^1 \int_0^1 q(\pi) \theta d\rho(\theta|\pi)dF_I(\pi)\\
		& = \int_0^1 q(\pi)\int_0^1  \theta d\rho(\theta|\pi) dF_I(\pi) \\
		& = \int_0^1 q(\pi)  e_I(\pi)  dF_I(\pi),
	\end{align*}
	where the first equality follows from applying Fubinis theorem, and the second one from Bayes' rule and Fubini's theorem. \\

	Next, we show that $\nu(\mathcal{M})=0$ for any mechanism if and only if $F=F_I$ and $\rho$ is unbiased. Note first that it follows trivially from representation given in \eqref{eqn:vos} that the value of attention is zero if $F_I=F$ and $\rho$ is unbiased. For the other direction, suppose that $F \neq F_I$ or the PGP $\rho$ is biased. We show that there exist allocation rule $q$ so that $\nu(q)>0$. Recall that $\nu(q )= \int_0^1 \Big( U(\theta|\theta)-\int_0^1 U(\pi|\theta) \Big) d\rho(\pi|\theta)dF(\theta)$. Take $q(x)=x$. Incentive compatibility and a bit of algebra establish that $  U(\theta|\theta)-\int_0^1 U(\pi|\theta)d\rho(\pi|\theta)= \int_0^1 q(\theta)(\theta-\pi)-\int_{\pi}^{\theta} q(s)ds d\rho(\pi|\theta) =\int_0^1 \theta^2/2-\pi^2/2- \pi(\theta-\pi) d\rho(\pi|\theta) = \int_0^1 (\theta -\pi)^2/2   d\rho(\pi|\theta)$. Thus, the value of attention becomes
	$$\nu(q )= \int_{(\theta,\pi) \in [0,1]^2} (\theta-\pi)^2/2 \quad d\mu(\theta,\pi) = \int_{(\theta,\pi):(\theta \neq \pi)} (\theta-\pi)^2/2 \quad d\mu(\theta,\pi),$$
	where $\mu(\theta,\pi)\equiv \int_0^1 \int_0^1 d\rho(\pi|\theta)dF(\theta)=   \int_0^1 \int_0^1 d\rho(\pi|\theta)dF_I(\pi)$, by Bayes' rule. Whenever $F\neq F_I$ or the PGP is biased, the measure (under $\mu$) of type-perceptions for which $\theta \neq \pi$ is strictly positive, so that the last integral is strictly positive.\\

	Finally, fix a non-constant allocation rule $q$. We show that there exists a PGP $\rho$ so that $\nu(q )>0$. Consider $d\rho(0|\theta)=1/2=d\rho(1|\theta)$, i.e., the perception maps each type $\theta$ to 0 or 1. In this case,
	$U(\theta|\theta)-\int_0^1 U(\pi|\theta)d\rho(\pi|\theta) $
	$=1/2 \Big( \int_0^{\theta} q(s)-q(0) ds + \int_{\theta}^1 1-q(s) ds \Big)$. Note that each element in the last bracket is non-negative as $q$ is non-decreasing. Moreover, if $q$ is non-constant everywhere, at least one of the integrals is strictly positive. Thus, for any non-constant $q$, we have that $ \nu(q )= \int_0^1   1/2 \Big( \int_0^{\theta} q(s)-q(0) ds + \int_{\theta}^1 1-q(s) ds \Big)   dF(\theta) >0$.

\end{proof}

  Before proving Proposition~\ref{prop:coarsescreening}, we describe some technical details that aid in its proof. To this end, we first formally define threshold mechanisms.
\begin{appxdef}\label{defn:threshold}
    Allocation rule $q:[0, 1]\to [0, 1]$ is a \textbf{threshold allocation rule} 
    if there exists $\pi^* \in [0, 1]$ such that $q(\pi) =1$ if $\pi > \pi^*$, $q(\pi) = 0$ if $\pi < \pi^*$, and $q(\pi^*) \in \{0, 1\}$.
\end{appxdef}

Let $\tau$ denote the set of threshold allocation rules. Recall that the set of feasible allocation rules in our problem, $Q$, consists of any non-decreasing $q:[0, 1] \to [0, 1]$. As such, the set of threshold allocation rules are those non-decreasing functions that map every perception to the extreme allocations 0 or 1 and, thus, are the set of extreme points of $Q$ (see \cite{Kleiner2021} or \cite{Yang2023}).\footnote{For a set $A$, $x \in A$ is an extreme point of $A$ if $x = \alpha y + (1 - \alpha)z$ for $y, z \in A$ and $\alpha \in [0, 1]$ imply that either $x = y$ or $x = z$.} Let $\tau^* \subset \tau$ denote the set of threshold allocation rules that maximize the value of attention, and note that $\tau^*$ is non-empty.  
Moreover, let $Q^*$ denote the set of all allocation rules that maximize the value of attention (so that $\tau^* \subseteq Q^*$). The following proposition characterizes the set $Q^*$.

\begin{appxprop}\label{prop:max_characterization}
   It holds that $q \in Q^*$ if and only if there exists a probability measure $\lambda$ supported on $\tau^*$ such that
            $$
        q = \int \limits_{\tau^*} q^\prime d\lambda(q^\prime). 
        $$
\end{appxprop}

Proposition~\ref{prop:max_characterization} implies that the set of threshold maximizers are sufficient for identifying all maximizers of $\nu$. That is, any $q \in Q^*$ is, essentially, a ``convex-combination" of those threshold mechanisms that maximize the value of sophistication.\footnote{Formally, this ``convex-combination" is a Bochner integral. See Footnote~14 of \cite{Kleiner2021} and the reference therein.}\\

\textbf{Proof of Proposition~\ref{prop:max_characterization}}

\begin{proof}
	Note that $\nu(\mathcal{M})$ is a linear functional in $q$, and consider the problem of maximizing it over $Q$, the set of all monotonic $q:[0,1]\rightarrow [0,1]$. It follows from Proposition 1, part 1, in \cite{Kleiner2021} that this set is convex and compact in the norm topology, and that Bauer's Maximum principle applies and implies that the functional attains a maximum on one of $Q$'s extreme points.   
	Moreover, \cite{Yang2023} implies that the set of $Q$'s extreme points are precisely the threshold mechanism introduced in Definition~\ref{defn:threshold}, the set $\tau$. 
	This establishes that the set $\tau^*$ is non-empty. \\
	
	We next show that $\tau^*$ is the set of extreme points of the set $Q^*$ (the set of all maximizers of our linear functional).
	Let $B(Q^*)$ denote the set of extreme points in $Q^*$.
	Take any $q^*\in B(Q^*)$ and suppose, by way of contradiction, that $q^* \not \in \tau^*$.
	Note that this implies that $q^* \not \in \tau$, that is, $q^*$ is not an extreme point of $Q$. This, together with the hypothesis that $q^*$ is an extreme point of $B(Q^*)$ implies that there is $q_1$,  $q_2$ (with $q_1\neq q_2$) and $  q^*  \in Q \setminus B(Q^*)$ and $\alpha \in (0,1)$ such that $q^*=\alpha q_1 + (1-\alpha)q_2$.
	Then consider the value of our linear functional: $\nu(q^*)=\nu(\alpha q_1+(1-\alpha )q_2)= \alpha \nu(q_1)+(1-\alpha) \nu (q_2)$. Note that this value is necessarily below its optimized solution value as $q_1$ and $q_2$ are not in $Q^*$, contradicting the hypothesis that $q^* \in Q^*$.\\
	
	Finally, having established that the set of extreme points of $Q^*$ is $\tau^*$, the integral representation is an immediate consequence of Proposition 1, part 2, of \cite{Kleiner2021}.

\end{proof}

 \textbf{Proof of Proposition~\ref{prop:coarsescreening}} 
\begin{proof}
    By Proposition~\ref{prop:gen Z}, part (a) there is $\nu(q)>0$ for some $q \in Q$. Thus, $\tau^* \ne \tau$ as otherwise Proposition~\ref{prop:max_characterization} would imply all $q \in Q$ are maximizers, requiring $\nu(q) = 0$. This also implies that the sets $\overline{\Pi}\equiv \lbrace \pi: q(\pi)=1 \quad \forall q \in \tau^* \rbrace$ and $\underline{\Pi}\equiv \lbrace x: q(\pi)=0 \quad \forall q \in \tau^* \rbrace$ are non-empty. Note that both $\overline{\Pi}$ and $\underline{\Pi}$ are, by definition, disjoint and both are intervals. To see this for $\overline{\Pi}$ (the proof is similar for $\underline{\Pi})$, take any $\pi < \pi^\prime$ with $\pi \in \overline{\Pi}$. Then, $q(\pi) = 1$ for all $q \in \tau^*$ and, since $q$ is non-decreasing, $q(\pi^\prime) = 1$ for all $q \in \tau$ and thus $\pi^{\prime} \in \overline{\Pi}$. Since by Proposition~\ref{prop:max_characterization}, $q \in Q^*$ if and only if there exists probability measure $\lambda$ such that $q = \int \limits_{\tau^*} q^\prime d \lambda (q^\prime)$, it follows immediately that $q(\pi) = 0$ for any $q \in Q^*$ if and only if $\pi \in \underline{\Pi}$ and $q(\pi) = 1$ for any $q \in Q^*$ if and only if $\pi \in \overline{\Pi}$.
     
\end{proof}
\textbf{Proof of Proposition~\ref{prop:BlackwellNew}}

\begin{proof}

Before proving the equivalence of the statements, we first introduce Lemma~\ref{lem:l1}, establishing that a smaller inattentive utility is equivalent to an order on the statistic $S(\cdot;\rho) $. By a small abuse of notation, let $\nu(q;\rho)$ and $V_I(q;\rho)$ be an agent's (with PGP $\rho$) value of attention and inattentive utility from mechanism with allocation $q \in Q$. To utilize on space, throughout the proof we make use of the following observation:

\begin{observation}\label{o1}
For any $q \in Q$, it holds that $V_I(q;\rho)-V_I(q;\rho')= \nu(  q; \rho')-\nu(  q;\rho')$.
\end{observation}
\begin{proof}
    Add and subtract $\int_0^1 q(\pi)(1-F(\pi))d\pi $ to the differences in inattentive utilities.
\end{proof}
\begin{appxlem}{\label{lem:l1}}
Fix a prior $F$ and two PGPs, $\rho$ and $\rho'$. Then, $V_I(q,\rho) \geq V_I(q,\rho') $ for all mechanisms $ q \in Q$ if and only if $S(x,\rho) \leq S(x,\rho') $ for all $ x \in [0,1]$.
\end{appxlem}

\begin{proof}

Suppose first that $V_I(q,\rho) \geq V_I(q,\rho') $ for all $ \mathcal{M}$ with associated allocation rule $q \in Q$. By Observation~\ref{o1} the hypothesis implies that $\nu(q;\rho)\leq \nu(q;\rho')$ $\forall q \in Q$, and, in particular for any threshold mechanisms with allocation rule the form of $q(\pi)=1$ if $\pi>x$ and zero else. For any such $q$, using the representation given in Proposition~\ref{prop:gen Z}, we see that $\nu(q;\rho) -\nu(q;\rho')=S(x;\rho')-S(x;\rho)$. Since the (first) difference is positive for any threshold allocation rule (i.e., $x \in [0,1]$), we have that $S(x;\rho')\geq S(x;\rho) $ for any $x \in [0,1]$.

Suppose next that $S(x;\rho)\leq S(x;\rho')$ for all $x \in [0,1]$. As $S(x;\rho)-S(x;\rho')= \nu(\tilde q; \rho)-\nu(\tilde q;\rho')$ for a threshold mechanism with allocation rule $\tilde q(\pi)=1$ if $\pi>x$ and zero else, assume without loss of generality that there is $\check x \in [0,1]$ such that $\nu(\check q;\rho)> \nu(\check q;\rho')$ with $\check q(\pi)=1$ if $\pi \geq \check x$ and zero else. By way of contradiction, take $\epsilon>0$, and observe that $\lim \limits_{\epsilon \to 0} (S(\check x-\epsilon;\rho)-S(\check x - \epsilon;\rho'))\rightarrow \nu(\check  q;\rho)-\nu(\check q; \rho'))$. Since $S(\check x-\epsilon;\rho)-S(\check x - \epsilon;\rho')\leq 0$ for all $\epsilon >0$, the inequality is preserved in the limit, and we arrive at a contradiction.

Finally, we generalize the conclusion to any $q \in Q$: Define linear functional $\Delta_{\rho,\rho'} \nu (q) \equiv \nu(q;\rho)-\nu(q;\rho') $ for $q \in Q$. It follows from Proposition~\ref{prop:max_characterization} that this linear functional has a maximizer in the set of extreme points of its domain (which is $\tau$). That is, there is a maximizer, say $q_{\max}$, such that, $q_{\max} \in \tau$. Moreover, $\Delta_{\rho,\rho'} \nu(q_{\max}) \leq 0$ by hypothesis. Thus, by Observation~\ref{o1}, for any $q \in Q$ we have $V_I(q,\rho')-V_I(q,\rho)=\Delta_{\rho,\rho'} \nu(q) \leq \Delta_{\rho,\rho'} \nu(q_{\max}) \leq 0 $.
    
\end{proof}

We are now in a position to prove the statement of the proposition.

$1. \Rightarrow 2.$: Suppose $\rho$ is more accurate than $\rho'$. By definition, the agent's welfare 
\begin{equation}\label{eq:AWA}
    \max \Big \lbrace \int_0^1 U(\theta|\theta)dF(\theta)-\kappa, \int_0^1 \int_0^1 U(\pi|\theta)d\rho(\pi|\theta)dF(\theta) \Big \rbrace
\end{equation}  is larger under PGP $\rho$ than under PGP $\rho'$ for any $q \in Q$ and $\kappa$. Considering $\kappa=\infty$, it is straightforward to observe that $V_I(q;\rho) \geq V_I(q;\rho')$ for any $q \in Q$. Invoking Lemma~\ref{lem:l1}, the result follows.

$2. \Rightarrow 1.$: Suppose $S(x;\rho) \leq S(x,\rho')$ for all $x \in [0,1]$. By Lemma~\ref{lem:l1} we have that $V_I(q;\rho) \leq V_I(q;\rho')$ for all $q \in Q$.
	Note that the first element of the agent's welfare (see \eqref{eq:AWA} above) is independent of the agent's PGP while the second element is larger under PGP $\rho$ since $V_I(q;\rho) \leq V_I(q;\rho')$ for all $q \in Q$ by hypothesis. Thus, the agent's welfare is larger under PGP $\rho$ than under $\rho'$ meaning that $\rho$ is more accurate than $\rho'$.

\end{proof}

\textbf{Proof of Proposition~\ref{prop:welfare_exo}}

\begin{proof}
	The first statement of the proposition follows trivially from comparing the selling the firm to the agent mechanism to the upper bounds on respective welfare.\\
	
	For the second statement, suppose the cost function is $C(x)=x^2/2$, and consider an agent with biased PGP operating under inattention.

	Under selling the firm to the agent, we have that $q(\pi)=\pi$, so that expected welfare is
	$\int_0^1 \pi(e_I(\pi)-\pi/2) dF_I(\pi)$.  
	Efficiency for inattention, however, requires that
	$q(\pi)=e_I(\pi)$, leading to expected welfare of $\int_0^1 e_I(\pi)^2/2 dF_I(\pi)$. Subtracting the solution value of efficiency under inattention from the welfare induced by selling the firm, and applying some algebra, we see that their difference is $\int_0^1 -(\pi-e_I(\pi))^2/2 dF_I(\pi)<0 $ for any biased PGP.

\end{proof}

\textbf{Proof of Proposition~\ref{prop:manage}}

\begin{proof}
	First, assume that $e_I$ is non-decreasing. Then, it is easy to see that manage the process induces a non-decreasing allocation rule. Thus, the manage the process mechanism is incentive compatible and feasible. Moreover, it follows from its definition that it achieves the upper bound on welfare for inattention.
	
	Next, suppose that $e_I$ is not non-decreasing, and consider cost function $C(x)=x^2/2$. Any mechanism that is efficient under inattention requires that $q(\pi)=e_I(\pi)$ almost everywhere on $supp(F_I)$, implying a non non-decreasing allocation rule. Thus, such a mechanism is not incentive compatible and, thus, not feasible.
\end{proof}

	\textbf{Proof of Proposition~\ref{prop:jusso}}
	
	\begin{proof}
		Suppose the agent faces the sell the firm to the agent mechanism and operates under inattention.
		In this case, they will choose $q(\pi) \in \arg \max_x \pi x- C(x)$. As the PGP is unbiased, the expected value from behaving under inattention is $W_I^*=\int_0^1 \max_x \lbrace e_I(\pi) q(x)-C(q(x)) \rbrace dF_I(\pi)$. If attentive, the agent's expected utility from behaving optimally is $\int_0^1 \max_x \lbrace \pi q(x)-C(x) \rbrace dF(\pi)-\kappa $. Thus, the agent's value of attention from the mechanism is $W^*_A-W^*_I$, and they operate under attention, only if $W^*_A-\kappa \geq W^*_I$, as prescribed by efficiency. Finally, it is easy to see that the agent's welfare, $\max \lbrace W_A^*-\kappa, W_I^* \rbrace$, is non-increasing in $\kappa$. 
	\end{proof}
	
	\textbf{Proof of Proposition~\ref{prop:inefficiency}}
	
	\begin{proof}
	In the following, we characterize the optimal mechanism on a case-by-case basis. To do so, the following two lemmas will be helpful.
		\begin{appxlem}\label{lem:sa}
			Suppose the designer implements selling the firm to the agent. Then, the agent operates under attention if it is optimal for welfare given the choice of the mechanism. 
				\end{appxlem}
				\begin{proof}
					Because the mechanism is selling the firm to the agent, it is straightforward to observe that the agent's (rational) expected utility coincides with welfare. Thus, if the agent operates under inattention, their inattentive utility coincides with welfare induced by inattentive behavior, say $W_I^e$.
					Therefore, the agent's value of attention is $W_A^*-W_I^e$, and the agent operates under attention only if $W_A^*-\kappa \geq W_I^e$.
				\end{proof}
 \begin{appxlem}\label{efifc}
                    If the principal wants to achieve attention, selling the firm is an optimal mechanism.
                \end{appxlem}
                \begin{proof}
                    By hypothesis, the optimal mechanism features allocation rule $q^*$ with $\nu(q^*)>\kappa $, so that the agent operates under attention.
				
				As selling the firm to the agent is efficient under attention, we only need to show that it provides the agent with sufficient attention incentives. Indeed, by Lemma~\ref{lem:sa}, the agent operates under attention if it is optimal for welfare to do so under the mechanism selling the firm to the agent. If it is welfare optimal to do so, total welfare equals $W_A^*-\kappa$; a weak upper bound on welfare of an attentive agent with cognitive costs $\kappa$. This value is necessarily larger than the welfare implied by inattentive behavior, say $W_I^e$, as otherwise the hypothesis that the optimal mechanism provides the agent with attention incentives was violated.
                \end{proof}
                
				We now characterize the optimal mechanism on a case-by-case basis.
				First, if $\kappa<\kappa^*$, the best the designer can hope for is that the agent operates under attention. By Lemma~\ref{lem:sa}, the agent operates under attention when choosing the mechanism selling the firm to the agent. 
				
				Second, if $\kappa \geq \kappa_I$, a managing the process mechanism achieves efficiency for inattention, and, by definition of $\kappa_I$, the agent operates under inattention. If $\kappa^* \geq \kappa_I$, then whenever efficiency for inattention is preferred to to efficiency under attention with aligned attention costs (i.e., $\kappa \geq \kappa^*)$, the agent operates under inattention as $\kappa \geq \kappa_I$.

				Finally, suppose that $\kappa_I>\kappa^*$ and $  \kappa \in ( \kappa^*,\kappa_I)$. We argue that there is $\bar \kappa  \in ( \kappa^*,\kappa_I)$ so that the solution to the principal's problem-- $\max_{q, \alpha \in [0, 1]}  \alpha \times  \left\{ \int \limits_0^1 e_I(\pi) q(\pi)-C(q(\pi)) d F_I(\pi) \right\} +(1-\alpha) \times \left\{ \int \limits_0^1 \pi q(\pi)-C(q(\pi)) d F(\pi) - \kappa   \right\}, $ such that $ \alpha(\nu(q)-\kappa) \geq 0$ and $(\alpha-1)(\nu-\kappa)\geq 0$-- is such that there exists $\bar \kappa$ satisfying the following properties: if $\kappa <\bar \kappa$, selling the firm to the agent is optimal and the agent is attentive (i.e., $\alpha=1$). If $\kappa >\bar \kappa$, the principal introduces distortions to the allocation rule in order to keep the agent inattentive (i.e., $\alpha =0$). Note that Berge's theorem of the maximum applies (as both the constraints and the objective of the principal's problem are continuous in the choice variables and the parameter $\kappa$), implying that the solution value varies continuously with $\kappa$.

                We first assume that the principal needs to keep the agent attentive or inattentive and solve the optimal mechanism for each scenario. The optimal mechanism then selects the maximum of both. By Lemma~\ref{efifc}, when holding the agent attentive, the principal does so by selling the firm to the agent.
Note that the solution value from selling the firm to the agent, $W_A^*-\kappa$, decreases in $\kappa$ on $(\kappa^*,\kappa_I)$ while attaining it's maximum, $W_I^*$, at $\kappa=\kappa^*$.

                Next, suppose that the principal keeps the agent inattentive. That is, the principal maximizes welfare conditional on the constraint that $\nu(q)-\kappa<0$. This constraint necessarily binds as $\kappa \leq \kappa_I$. Moreover, the larger $\kappa$, the less the constraint binds as in the limit $\nu(q_I^*)=\kappa_I$. Thus, the solution value from keeping the agent inattentive increases continuously (by Berge's theorem of the maximum) in $\kappa$ over $(\kappa^*,\kappa_I)$ and attains it's maximum, $W_I^*$ ar $\kappa=\kappa_I$.

Thus, by the intermediate value theorem, the solution value of the principal's problem when keeping the agent attentive and that when keeping the agent inattentive cross exactly once at some  $\kappa \in (\kappa^*,\kappa_I)$, which we call $\bar \kappa$.

			\end{proof}
            
			\textbf{Proof of Proposition~\ref{prop:carrotstick}}
            \begin{proof}
            The proof centers around two lemmas. Lemma~\ref{lemma:unbiased} characterizes the optimal allocation rule when the agent's PGP is unbiased (i.e., $\rho_U$), while Lemma~\ref{lemma:biased} characterizes that for the case of an unbiased (i.e., $\rho_C)$. After proving each lemma, we use these allocation rules to derive the respective comparative static results on $V_A$ and $V_I$.
            
			\begin{appxlem}\label{lemma:unbiased}
	Suppose the agent's PGP is $\rho_U$. There exist a thresholds $0 < \underline{\kappa}_U < \overline{\kappa}_U$ such that the optimal mechanism, $q^*_U$, is
	\begin{enumerate}[(a)]
		\item if $\kappa \le \underline{\kappa}_U$, $q^*_U = q_A$ and the agent is attentive;
		\item if $\kappa \in (\underline{\kappa}_U,\overline{\kappa}_U)$,
		$$
		q^*_U(\pi) = \begin{cases}
			0 & \text{if $\pi \le \frac{1}{2}$}\\
			(2 + \lambda)\pi - (1 + \lambda/2) & \text{if $\pi \in (1/2, 3/4)$}\\
			(2 - \lambda)\pi - (1 - \lambda) & \text{if $\pi \ge 3/4$}
		\end{cases}
		$$
		where $\lambda = 96\kappa - 3$ and the agent is attentive; and
		\item if $\kappa \ge \overline{\kappa}_U$, $q^*_U = q_I$ and the agent is inattentive.
	\end{enumerate}
\end{appxlem}

			\begin{proof}
				Recall from the main text that profits are higher when the agent is attentive (with allocation rule $q_A$) versus when inattentive (with allocation rule $q_I$). With $\rho_U$, the value of attention given allocation rule $q$ is $\nu^U(q) = \int \limits_0^{1/4} q(\pi)[-\pi]d\pi + \int \limits_{1/4}^{3/4} q(\pi)\left[\pi - \frac{1}{2}\right] d\pi + \int \limits_{3/4}^1 q(\pi)[1 - \pi].
				$
				Thus, $\nu^U(q_A) = 1/32$ and $q_A$ is optimal if $\kappa \le 1/32 \equiv \underline{\kappa}_U$. If $\kappa > 1/32$, the designer compares the optimal mechanism that induces each cognitive state. The optimal mechanism that induces attention solves $\max \limits_{q} \int \limits_{0}^1 q(\pi)(2\pi - 1) - \frac{q(\pi)^2}{2}$, subject to $\nu^U(q) = \kappa$. If $\lambda$ is the Lagrange multiplier on the constraint, the optimal solution takes the form as stated in part (b) of the lemma. Substituting this into the constraint that $\nu^U(q) = \kappa$, we get that $\lambda = 96\kappa - 3$. We compare this solution to the optimal mechanism that induces inattention, and note that profits from inattention are bounded above by 9/128, which are achieved by $q_I$. Instead, profits under the optimal mechanism that induces attention are $1/12 - (96\kappa - 3)^2/192$, which is higher than 9/128 if and only if $\kappa \le \sqrt{5/2}/96 + 1/32 \equiv \overline{\kappa}_U$. Hence, this mechanism is optimal for $\kappa \in (\underline{\kappa}_U, \overline{\kappa}_U)$. 
				
				Finally, note that the lowest value of attention achievable by an inattentive-optimal mechanism is $\nu^U(q_I) = 21/512$ (done by setting $q_I(\pi) = 3/4$ for all $\pi > 3/4$). Since $21/512 < \overline{\kappa}_U$, it follows that the inattentive-optimal mechanism $q_I$ induces inattention for $\kappa \ge \overline{\kappa}_U$ and, as such, is the optimal mechanism. 
			\end{proof}
Using Lemma~\ref{lemma:unbiased} it is straightforward to observe that $V_A(q_U^*)=\int_{1/2}^1 q_U^*(\pi)(1- \pi)d\pi$ and $V_I(q_U^*)=V_A(q_U^*)-\nu^U(q_U^*)=\int_{1/2}^{3/4} q_U^*(\pi)(3/2-2\pi)d\pi$ increase in $\lambda$ (and thus in $\kappa$) on $\kappa \in (\underline{\kappa}_U, \bar{\kappa}_U)$.
            \begin{appxlem}\label{lemma:biased}
	Suppose the agent's PGP is $\rho_C$. There exist a thresholds $0 < \underline{\kappa}_C < \overline{\kappa}_C$ such that the optimal mechanism, $q^*_C$, is
	\begin{enumerate}[(a)]
		\item if $\kappa \le \underline{\kappa}_C$, $q^*_C = q_A$ and the agent is attentive;
		\item if $\kappa \in (\underline{\kappa}_C,\overline{\kappa}_C)$,
		$$
		q^*_C(\pi) = \begin{cases}
			0 & \text{if $\pi \le \frac{1}{2}$}\\
			(2 - \lambda)\pi - (1 - \lambda/2) & \text{if $\pi \in (1/2, 3/4)$}\\
			(2 - \lambda)\pi - (1 - \lambda) & \text{if $\pi \ge 3/4$}
		\end{cases}
		$$
		where $\lambda = 96\kappa - 1$ and the agent is attentive; and
		\item if $\kappa \ge \overline{\kappa}_U$, $q^*_U = q_I$ and the agent is inattentive.
	\end{enumerate}
\end{appxlem}
		 
			\begin{proof}
				Recall that profits are higher when the agent is attentive (with allocation rule $q_A$) versus when inattentive (with allocation rule $q_I$). With $\rho_C$, the value of attention given allocation rule $q$ is $\nu^C(q) = \int \limits_0^{1/4} q(\pi)[-\pi]d\pi + \int \limits_{1/4}^{3/4} q(\pi)\left[\frac{1}{2} - \pi\right] d\pi + \int \limits_{3/4}^1 q(\pi)[1 - \pi].
				$ 
				Thus, $\nu^C(q_A) = 1/96$ and $q_A$ is optimal if $\kappa \le 1/96 \equiv \underline{\kappa}_U$. If $\kappa > 1/96$, the designer compares the optimal mechanism that induces each cognitive state. The optimal mechanism that induces attention solves $ 
				\max \limits_{q} \int \limits_{0}^1 q(\pi)(2\pi - 1) - \frac{q(\pi)^2}{2},
				$ 
				subject to $\nu^C(q) = \kappa$. If $\lambda$ is the Lagrange multiplier on the constraint, the optimal solution takes the form as stated in part (b) of this lemma. Substituting this into the constraint that $\nu^U(q) = \kappa$, we get that $\lambda = 96\kappa - 1$. We compare this number to the generated profits of the optimal mechanism that induces inattention, which are bounded from below by  9/128 and implemented by $q_I$. Instead, profits under the optimal mechanism that induces attention are $1/12 - (96\kappa - 1)^2/192$, which is higher than 9/128 if and only if $\kappa \le \sqrt{5/2}/96 + 1/96 \equiv \overline{\kappa}_C$. Hence, this mechanism is optimal for $\kappa \in (\underline{\kappa}_C, \overline{\kappa}_C)$.

				Further note that the lowest value of attention achievable by an inattentive-optimal mechanism is $\nu^C(q_I) = 3/512$ (done by setting $q_I(\pi) = 3/4$ for all $\pi > 3/4$). Since $3/512 < \overline{\kappa}_C$, it follows that the inattentive-optimal mechanism $q_I$ induces inattention for $\kappa \ge \overline{\kappa}_U$ and, as such, is the optimal mechanism. 
			\end{proof}
Finally, consider $\kappa \in (\underline{\kappa}_C, \bar{\kappa}_C)$. Lemma~\ref{lemma:biased} implies that $dV_A(q_C^*)/d\lambda=\int_{1/2}^{3/4} (1/2-\pi)(1- \pi)d\pi+\int_{3/4}^1 (1-\pi)^2 d\pi<0$ and $dV_I(q_C^*)/d\lambda=d\Big(V_A(q_C^*) -\nu^C(q_C^*)\Big)/d\lambda=\int_{1/2}^{3/4} dq_U^*(\pi)/\lambda (1/2)d\pi=\int_{1/2}^{3/4}  (1/2-\pi)/2 d\pi<0$. Thus, both $V_A$ and $V_I$ decrease in $\lambda$ (and thus in $\kappa$) on $(\underline{\kappa}_C, \bar{\kappa}_C)$. 
            \end{proof}

			\textbf{Proof of Proposition~\ref{prop:optimal_hype}}
\begin{proof}
            
			We first derive the seller's optimal price for any $\kappa$ and level of hype $h$ (Lemma~\ref{lemma:optimalprice1}). Using this pricing behavior, we then characterize the seller's and buyer's preferred level of hype (Lemma~\ref{prop:optimal_hype1}).

             Recall that the seller's optimal price under exogenous inattention is $p_I(h)=\min \lbrace 1, (2(1-h))^{-1} \rbrace$, and, for fixed $h$, define $\bar \kappa(h) \equiv \frac{1}{2} h p_I(h)^2  $ (the black curve in Figure~\ref{fig:hype}) and $\underline{\kappa}(h)$ (the blue curve in Figure~\ref{fig:hype}) as solution to $\left[1 - (1 - h)\left(\frac{2\kappa}{h}\right)^{1/2}\right]\left(\frac{2\kappa}{h}\right)^{1/2} =\frac{1}{4}$. Note that $\underline{\kappa}(h)<\bar \kappa(h)$. The following lemma describes the seller's optimal price for any $h$ and $\kappa$.
\begin{appxlem}\label{lemma:optimalprice1}
	The seller's optimal price, $p^*(h)$, is given by
	$$
	p^*(h) = \begin{cases}
		p_I(h) & \text{if $\kappa > \overline{\kappa}(h)$}\\
		\left(\frac{2\kappa}{h}\right)^{1/2} &  \text{if $\kappa \in [\underline{\kappa}(h), \overline{\kappa}(h)]$} \\
		\frac{1}{2} & \text{if $\kappa < \underline{\kappa}(h)$}.
	\end{cases} 
	$$
	The consumer is inattentive if $\kappa \ge \underline{\kappa}(h)$ and attentive if $\kappa < \underline{\kappa}(h)$.
\end{appxlem}

\begin{proof}
				The seller's most preferred outcome is to sell to an inattentive buyer at a price of $p_I(h)$. Moreover, if $\kappa \ge \nu(p_I(h))= \frac{1}{2} h p_I(h)^2 =\overline{\kappa}(h)$, a price of $p_I(h)$ induces inattention.\\
				
				If $\kappa < \overline{\kappa}(h)$, the seller needs to compare the optimal price that induces inattention to that which induces attention. For inattention, since the buyer's value of attention is decreasing in $p$ and profits are convex in $p$ (maximized at $p_I(h)$), it follows that the optimal price that induces inattention is implicitly defined by $\nu(p) = \kappa$ or $p = (2\kappa/h)^{1/2}$. The agent becomes attentive for any price above $(2\kappa/h)^{1/2}$ and, since the $1/2$ is the unconstrained optimal price to offer an attentive buyer, the optimal price that induces attention is $\max\left\{(2\kappa/h)^{1/2}, 1/2\right\}$. Since inattention implies a higher probability of sale for a given price, it follows that the designer chooses $(2\kappa/h)^{1/2}$ to induce inattention as long us
				\begin{equation}\label{eqn:lowestk}
					\left[1 - (1 - h)\left(\frac{2\kappa}{h}\right)^{1/2}\right]\left(\frac{2\kappa}{h}\right)^{1/2} \ge \frac{1}{4},
				\end{equation}
				where the left-hand side of the inequality increases in $\kappa$. Thus, by definition, $\underline{\kappa}(h)$ is the smallest $\kappa$ such that \eqref{eqn:lowestk} holds. Moreover, note that $(2\underline{\kappa}/h)^{1/2} < 1/2$. Hence, $p^*(h) = (2\kappa/h)^{1/2}$, inducing inattention, for $\kappa \in [\underline{\kappa}(h), \overline{\kappa}(h)]$ and $p^*(h) = 1/2$, inducing attention, for $\kappa < \underline{\kappa}(h)$.
			\end{proof}

            \begin{appxlem}\label{prop:optimal_hype1}
	The seller-optimal level of hype, $h_s^*$, and the buyer-optimal level of hype, $h_b^*$, are
	$$
	h_s^* = \begin{cases}
		1 & \text{if $\kappa \ge 1/8$}\\
		8\kappa & \text{if $\kappa < 1/8$}
	\end{cases};\quad
	h_b^* = \begin{cases}
		0 & \text{if $\kappa \ge 9/128$}\\
		1 & \text{if $\kappa \in (1/32, 9/128)$}\\
		\frac{16\kappa}{(\sqrt{2} - 4\sqrt{\kappa})^2} & \text{if $\kappa \le 1/32$}.
	\end{cases}
	$$
	If and only if $\kappa < 9/128$, the buyer strictly prefers a higher degree of hype than the seller. Moreover, for fixed $\kappa$ the seller's revenues and the buyer's utility vary with $h$ as depicted in Figure~\ref{fig:hype}.
     
            \end{appxlem}

            \begin{proof}
            We show the result on a case-by-case basis. First, consider $\kappa$ and $h$ such that $\kappa \ge \overline{\kappa}(h)$ (the area above the black dotted curve in Figure~\ref{fig:hype}). Then, this is equivalent to the case of exogenous inattention, so that $R^*(h)$ increases and, the agent's indirect utility as a function of $h$ and $p^*(h)$, say $V^*(h)$, decreases as $h$ increases. Next, suppose that $\kappa$ and $h$ are such that $\kappa < \underline{\kappa}(h)$ (the area below the blue solid curve in Figure~\ref{fig:hype}). In this case, the seller charges a price of $1/2$ and the buyer is attentive. Hence, $R^*(h) = 1/4$ and $V^*(h) = 1/8 - \kappa$, both of which are independent of $h$.\\
				
				Now, consider the intermediate case in which $\kappa$ and $h$ satisfy $\kappa \in [\underline{\kappa}(h), \overline{\kappa}(h)]$. In this case, the seller's revenue can be written as $[1 - (1 - h)p^*(h)]p^*(h)$ where $p^*(h) = (2\kappa/h)^{1/2}$. Using that the derivative of $p^*(h)$ in $h$ is $-(1/2h)p^*(h)$, the derivative of seller revenue in $h$ is $(1/h)p^*(h)[p^*(h) - 1/2]$, which is increasing if $p^*(h) > 1/2$ and decreasing if $p^*(h) < 1/2$. If $\kappa \ge 1/8$, then $p^*(1) \ge 1/2$. If $\kappa < 1/8$, then $p^*(8\kappa) = 1/2$. Since $\underline{\kappa}(h) < (1/8)h < \overline{\kappa}(h)$, it follows that the seller-preferred degree of hype is $h_s^* = \min\{1, 8\kappa\}$ and the seller's revenue decreases in $h$ for $h>h^*_s$ (i.e., to the right of the red dashed curve depicting $(1/8) h$ in Figure~\ref{fig:hype}) while it increases in $h$ for $h<h_s^*$ (i.e., to the left of that red dashed curve).\\
				
				The buyer's indirect utility takes the form $V^*(h) = 1/2 - p^*(h) + (1/2)(1 - h)p^*(h)^2$ if $\kappa$ and $h$ satisfy $\kappa \in [\underline{\kappa}(h), \overline{\kappa}(h)]$. Note that $V^*(h)$ has derivative equal to $\frac{1}{2h}p^*(h)\left[1 - (1 - h)p^*(h)\right] > 0$. 
				Hence, for fixed $\kappa$, the buyer either prefers $h = 0$ (as $V^*$ decreases in $h$ as long as $\kappa > \overline{\kappa}(h)$) or the largest feasible $h$ if $\kappa < \overline{\kappa}(h)$ (either $h = 1$ or the unique $h$ such that $\underline{\kappa}(h) = \kappa$). The former is preferred if $1/8 > 1/2 - (2\kappa)^{1/2}$ or $\kappa > 9/128$. If $\kappa < 9/128$ the latter case is preferred. This gives $h_b^* = 1$ if $\kappa \in (1/32, 9/128)$ and $h_b^*$ is the unique solution to $\underline{\kappa}(h_b^*) = \kappa$ for $\kappa < 1/32$ (i.e., the inverse of the solid blue curve in Figure~\ref{fig:hype}).

			\end{proof}

            \end{proof}
			
		\end{document}